\documentclass[11pt]{article}

\usepackage{xspace}             
\usepackage{hyperref}

\usepackage{amsmath}
\usepackage{amsthm}
\usepackage{amsfonts}
\usepackage{amssymb}

\usepackage{caption}
\usepackage{subcaption}
\usepackage{xcolor,cancel}

\usepackage{graphicx}           
\usepackage{import}             
\usepackage{epstopdf}           

\usepackage{enumerate}
\usepackage{algorithm}
\usepackage{algpseudocode}

\algnewcommand\algorithmicinput{\textbf{Input:}}
\algnewcommand\Input{\item[\algorithmicinput]}
\algnewcommand\algorithmicoutput{\textbf{Output:}}
\algnewcommand\Output{\item[\algorithmicoutput]}

\newtheorem{theorem}{Theorem}
\newtheorem{lemma}{Lemma}
\newtheorem{corollary}{Corollary}
\newtheorem{definition}{Definition}

\newtheorem{claim}{Claim}

\newcommand{\eps}{\epsilon}

\newcommand{\diag}{\mathtt{diag}}
\newcommand{\ai}[1]{\|a_{., #1}\|}
\newcommand{\ao}[1]{\|a_{#1,.}\|}
\newcommand{\ali}[1]{\|a_{., #1}\|_1}
\newcommand{\alo}[1]{\|a_{#1, .}\|_1}
\newcommand{\aij}{a_{ij}}
\newcommand{\sij}{\sum_{i,j}}
\newcommand{\x}{\mathbf{x}}
\newcommand{\xt}[1]{\mathbf{x^{(#1)}}}

\newcommand{\hxt}[1]{\widehat{\mathbf{x}}^{(#1)}}
\newcommand{\hx}{\hat{x}}
\newcommand{\hai}{\|\widehat{a}^{(t)}_{.,i}\|}
\newcommand{\khai}{\|\widehat{a}^{(t)}_{.,k}\|}
\newcommand{\hao}{\|\widehat{a}^{(t)}_{i,.}\|}
\newcommand{\khao}{\|\widehat{a}^{(t)}_{k,.}\|}
\newcommand{\haij}{\widehat{a}^{(t)}_{ij}}
\newcommand{\haji}{\widehat{a}^{(t)}_{ji}}

\newcommand{\dd}{\mathbf{d}}

\DeclareMathOperator*{\argmax}{arg\,max}

\newcommand{\RR}{{\mathbb{R}}}

\title{Matrix Balancing in $L_p$ Norms:\\
         A New Analysis of Osborne's Iteration}

\author{Rafail Ostrovsky\\UCLA\\rafail@cs.ucla.edu\thanks{Research supported in part by NSF grants 1065276, 1118126 and 1136174, US-Israel BSF grants, OKAWA Foundation Research Award, IBM Faculty Research Award, Xerox Faculty Research Award, B. John Garrick Foundation
Award, Teradata Research Award, and Lockheed-Martin Corporation Research Award.
This material is also based upon work supported in part by DARPA
Safeware program. The views expressed are those of the authors and do not reflect the official policy or position of the Department of Defense or the U.S. Government.} \and
Yuval Rabani\\The Hebrew University of Jerusalem\\yrabani@cs.huji.ac.il\thanks{
Research supported in part by ISF grant 956-15, by BSF grant 
2012333, and by I-CORE Algo.} \and
Arman Yousefi\\UCLA\\armany@cs.ucla.edu\footnotemark[1]
}
\begin{document}

\maketitle
\begin{abstract}
We study an iterative matrix conditioning algorithm due 
to Osborne (1960). The goal of the algorithm is to convert 
a square matrix into a {\em balanced} matrix where every 
row and corresponding column have the same norm. The 
original algorithm was proposed for balancing rows and 
columns in the $L_2$ norm, and it works by iterating over
balancing a row-column pair in fixed round-robin order.
Variants of the algorithm for other norms have been heavily 
studied and are implemented as standard preconditioners 
in many numerical linear algebra packages. Recently, Schulman
and Sinclair (2015), in a first result of its kind for any norm,
analyzed the rate of convergence of a variant of Osborne's 
algorithm that uses the $L_{\infty}$ norm and a different
order of choosing row-column pairs. In this paper we study 
matrix balancing in the $L_1$ norm and other $L_p$ norms. 
We show the following results for any matrix $A = (a_{ij})_{i,j=1}^n$, 
resolving in particular a main open problem mentioned by Schulman 
and Sinclair.
\begin{enumerate}
\item We analyze the iteration for the $L_1$ norm under a 
greedy order of balancing. We show that it converges to
an $\eps$-balanced matrix in
$K = O(\min\{\eps^{-2}\log w,\eps^{-1}n^{3/2}\log(w/\eps)\})$ 
iterations that cost a total of $O(m + Kn\log n)$ arithmetic 
operations over $O(n\log w)$-bit numbers.
Here $m$ is the number of non-zero entries of $A$,
and $w = \sum_{i,j} |a_{ij}|/a_{\min}$ with 
$a_{\min} = \min\{|a_{ij}|:\ a_{ij}\neq 0\}$.
\item We show that the original round-robin implementation
converges to an $\eps$-balanced matrix in 
$O(\eps^{-2}n^2\log w)$ 
iterations totalling $O(\eps^{-2}mn\log w)$ arithmetic 
operations over $O(n\log w)$-bit numbers.
\item We show that a random implementation of the iteration
converges to an $\eps$-balanced matrix in $O(\eps^{-2}\log w)$
iterations using $O(m + \eps^{-2}n\log w)$ arithmetric
operations over $O(\log(wn/\eps))$-bit numbers.
\item We demonstrate a lower bound of $\Omega(1/\sqrt{\eps})$ on 
the convergence rate of any implementation of the iteration.
\item We observe, through a known trivial reduction, that our results 
for $L_1$ balancing apply to any $L_p$ norm for all finite $p$, at the 
cost of increasing the number of iterations by only a factor of $p$.
\end{enumerate}
We note that our techniques are very different from those used
by Schulman and Sinclair.
\end{abstract}

\thispagestyle{empty}
\newpage
\setcounter{page}{1}

\section{Introduction}

Let $A=(a_{ij})_{n\times n}$ be a square matrix with real entries, and 
let $\|\cdot\|$ be a given norm. For an index $i\in[n]$, let $\ao{i}$ and 
$\ai{i}$, respectively, denote the norms of the $i$th row and the $i$th 
column of $A$, respectively. The matrix $A$ is \emph{balanced} in 
$\|\cdot\|$ iff $\ai{i} = \ao{i}$ for all $i$. An invertible diagonal matrix 
$D=\diag(d_1,\ldots, d_n)$ is said to \emph{balance} a matrix $A$ iff 
$DAD^{-1}$ is balanced. A matrix $A$ is \emph{balanceable} in 
$\|\cdot\|$ iff there exists a diagonal matrix $D$ that balances it. 

Osborne~\cite{osborne} studied the above problem in the $L_2$ norm 
and considered its application in preconditioning a given matrix in order 
to increase the accuracy of the computation of its eigenvalues. The 
motivation is that standard linear algebra algorithms that are used to 
compute eigenvalues are numerically unstable for unbalanced matrices; 
diagonal balancing addresses this issue by obtaining a balanced matrix 
that has the same eigenvalues as the original matrix, as $DAD^{-1}$ 
and $A$ have the same eigenvalues. Osborne suggested an iterative 
algorithm for finding a diagonal matrix $D$ that balances a matrix $A$, 
and also proved that his algorithm converges in the limit. He also observed 
that if a diagonal matrix $D=\diag(d_1,\ldots,d_n)$ balances a matrix 
$A$, then the diagonal vector $\mathbf{d} =(d_1, \ldots, d_n)$ minimizes 
the Frobenius norm of the matrix $DAD^{-1}$. Osborne's classic algorithm 
is an iteration that at each step balances a row and its corresponding 
column by scaling them appropriately. More specifically the algorithm 
balances row-column pairs in a fixed cyclic order. In order to balance 
row and column $i$, the algorithm scales the $i$th row by $\sqrt{\ai{i}/\ao{i}}$ 
and the $i$th column by $\sqrt{\ao{i}/\ai{i}}$. Osborne's algorithm converges 
to a unique balanced matrix, but there have been no upper bounds on the 
converges rate of Osborne's algorithm for the $L_2$ norm prior to our work.

Parlett and Reinsch~\cite{parlett} generalized Osborne's algorithm to 
other norms without proving convergence. The $L_1$ version of the 
algorithm has been studied extensively. The convergence in the limit 
of the $L_1$ version was proved by Grad~\cite{grad}, uniqueness of 
the balanced matrix by Hartfiel~\cite{hartfiel}, and a characterization 
of balanceable matrices was given by Eaves et al.~\cite{eaves}. Again, 
there have been no upper bounds on the running time of the $L_1$ 
version of the iteration. The first polynomial time algorithm for balancing 
a matrix in the $L_1$ norm was given by Kalantari, Khachiyan, and Shokoufandeh~\cite{khachiyan}. Their approach is different from the 
iterative algorithm of Osborne-Parlett-Reinsch. They reduce the balancing 
problem to a convex optimization problem and then solve that problem 
approximately using the ellipsoid algorithm. Their algorithm runs in 
$O(n^4 \log(n\log w/\eps))$ arithemtic operations where 
$w =\sum_{i,j}|a_{i,j}|/a_{\min}$ for $a_{\min}=\min\{|a_{ij}|: a_{ij}\neq 0\}$
and $\eps$ is the relative imbalance of the output matrix 
(see Definition~\ref{def:approx}).

For matrix balancing in the $L_{\infty}$ norm, 
Schneider and Schneider~\cite{schneider} gave an $O(n^4)$-time 
non-iterative algorithm. This running time was improved to 
$O(mn + n^2\log n)$ by Young, Tarjan, and Orlin~\cite{young91fasterparametric}. 
Despite the existence of polynomial time algorithms for balancing 
in the $L_1$ and $L_{\infty}$ norms, and the lack of any theoretical 
bounds on the running time of the Osborne-Parlett-Reinsch (OPR) iterative 
algorithm, the latter is favored in practice, and the Parlett and Reinsch
variant~\cite{parlett} is implemented as a standard in almost all linear 
algebra packages (see Chen~\cite[Section 3.1]{chenThesis}, also the
book~\cite[Chapter 11]{Numerical} and the
code in~\cite{EISPACK}). One reason is 
that iterative methods usually perform well in practice and run for far 
fewer iterations than are needed in the worst case. Another advantage 
of iterative algorithms is that they are simple, they provide steady partial 
progress, and they can always generate a matrix that is sufficiently balanced 
for the subsequent linear algebra computation. 

Motivated by the impact of the OPR algorithm 
and the lack of any theoretical bounds on its running time, Schulman and Sinclair~\cite{schulman} recently showed the first bound on the convergence 
rate of a modified version of this algorithm in the $L_{\infty}$ norm. They prove 
that their modified algorithm converges in $O(n^3\log (\rho n/\eps))$ balancing 
steps where $\rho$ measures the initial imbalance of $A$ and $\eps$ is the 
target imbalance of the output matrix. Their algorithm differs from the original 
algorithm only in the order of choosing row-column pairs to balance (we will
use the term {\em variant} to indicate a deviation from the original round-robin
order). Schulman and Sinclair 
do not prove any bounds on the running time of the algorithm for other $L_p$ 
norms; this was explicitly mentioned as an open problem. Notice that when
changing the norm, not only the target balancing condition changes but also
the iteration itself, so we cannot deduce an upper bound on the rate of
convergence in the $L_p$ norm from the rate of convergence in the $L_{\infty}$
norm.

In this paper we resolve the open question of~\cite{schulman}, and upper bound 
the convergence rate of the OPR iteration in any $L_p$ norm.\footnote{It should
be noted that the definition of target imbalance $\eps$ in~\cite{schulman} 
is stricter than the definition used by~\cite{khachiyan}. We use the definition
in~\cite{khachiyan}. This is justified by the fact that the numerical stability
of eigenvalue calculations depends on the Frobenius norm of the balanced
matrix, see~\cite{parlett}.}
Specifically, we show the following bounds for the $L_1$ norm. They imply
the same bounds with an extra factor of $p$ for the $L_p$ norm, by using
them on the matrix with entries raised to the power of $p$.
(Below, the $\tilde{O}(\cdot)$ notation hides factors that are logarithmic in
various parameters of the problem. Exact bounds await the statements
of the theorems in the following sections.)
We show that the original algorithm (with no modification) converges to an 
$\eps$-balanced matrix in $\tilde{O}(n^2/{\eps^2})$ balancing steps, using 
$\tilde{O}(mn/\eps^2)$ arithmetic operations. We also show that a greedy 
variant converges in $\tilde{O}(1/\eps^2)$ balancing steps, using 
$O(m) + \tilde{O}(n/\eps^2)$ arithmetic operations; or alternatively
in $\tilde{O}(n^{3/2}/\eps)$ iterations, using $\tilde{O}(n^{5/2}/\eps)$
arithmetic operations. Thus, the number of arithmetic operations needed 
by our greedy variant is nearly linear in $m$ or nearly linear in $1/\eps$. 
The near linear dependence on $m$ is significantly better than the 
Kalantari-Khachiyan-Shokoufandeh algorithm that uses 
$O(n^4 \log(n\log w/\eps))$ arithmetic operations (and also the
Schulman and Sinclair version with a stricter, yet $L_{\infty}$,
guarantee). For an accurate comparison
we should note that we may need to maintain $\tilde{O}(n)$ bits of precision, 
so the running time is actually $O(m + n^2\log n \log w/\eps^2)$ (the
Kalantari et al. algorithm maintains $O(\log(wn/\eps))$-bit numbers).
We improve this with yet another, randomized, variant that has similar convergence
rate (nearly linear in $m$), but needs only $O(\log(wn/\eps))$ bits of 
precision. Finally, we show that the dependence on $\eps$ given by our 
analyses is within the right ballpark---we demonstrate a lower bound of
$\Omega(1/\sqrt{\eps})$ on the convergence rate of any variant of the
algorithm to an $\eps$-balanced matrix. Notice the contrast 
with the Schulman-Sinclair upper bound for balancing in the $L_{\infty}$ 
norm that has $O(\log (1/\eps))$ dependence on $\eps$
(this lower bound is for the Kalantari et al. notion of 
balancing so it naturally applies also to strict balancing).


Osborne observed that a diagonal matrix $D=\diag(d_1,\ldots,d_n)$ 
that balances a matrix $A$ in the $L_2$ norm also minimizes 
the Frobenius norm of the matrix $DAD^{-1}$. Thus, the 
balancing problem can be reduced to minimizing a convex 
function. Kalantari et al.~\cite{khachiyan} gave a convex
program for balancing in the $L_1$ norm. Our analysis
is based on their convex program. We relate the OPR 
balancing step to the coordinate descent method in convex 
programming. We show that each step reduces the value of 
the objective function. Our various bounds are derived through 
analyzing the progress made in each step. In particular, one
of the main tools in our analysis is an upper bound on the
distance to optimality (measured by the convex objective
function) in terms of the the $L_1$ norm of the gradient,
which we prove using network flow arguments.

For lack of space, many proofs are missing inline.
They appear in Section~\ref{proofs}.
 


\section{Preliminaries}\label{sec:preliminaries}

In this section we introduce notation and definitions, we discuss 
some previously known facts and results, and we prove a couple 
of useful lemmas.

\paragraph{The problem.}
Let $A=(a_{ij})_{n\times n}$ be a square real matrix, and let $\|\cdot\|$ be 
a norm on $\RR^n$. For an index $i\in[n]$, let $\ao{i}$ and $\ai{i}$, respectively, 
denote the norms of the $i$th row and the $i$th column of $A$, respectively.
A matrix $A$ is \emph{balanced} in $\|\cdot\|$ iff $\ai{i} = \ao{i}$ for all $i$. 
An invertible diagonal matrix $D=\diag(d_1,\ldots, d_n)$ is said to \emph{balance} 
a matrix $A$ iff $DAD^{-1}$ is balanced. A matrix $A$ is \emph{balanceable} 
in $\|\cdot\|$ iff there exists a diagonal matrix $D$ that balances it. 

For balancing a matrix $A$ in the $L_p$ norm only the absolute values of the 
entries of $A$ matter, so we may assume without loss of generality that $A$ 
is non-negative.  Furthermore, balancing a matrix does not change its diagonal 
entries, so if a diagonal matrix $D$ balances $A$ with its diagonal entries replaced 
by zeroes, then $D$ balances $A$ too.  Thus, for the rest of the paper, we assume 
without loss of generality that the given $n\times n$ matrix $A=(a_{ij})$ is non-negative
and its diagonal entries are all $0$.

A diagonal matrix $D=\diag(d_1,\ldots,d_n)$ balances $A=(a_{ij})$ in the $L_p$ 
norm if and only if $D^p=\diag({d_1}^p,\ldots,{d_n}^p)$ balances the matrix 
$A'=({a_{ij}}^p)$ in the $L_1$ norm. Thus, the problem of balancing matrices in 
the $L_p$ norm (for any finitie $p$) reduces to the problem of balancing matrices 
in the $L_1$ norm; for the rest of the paper we focus on balancing matrices in the 
$L_1$ norm. 

For an $n\times n$ matrix $A$, we use $G_A = (V, E, w)$ to denote the weighted 
directed graph whose adjacency matrix is $A$. More formally, $G_A$ is defined
as follows. 
Put $V=\{1,\ldots, n\}$, put $E=\{(i,j): a_{ij} \neq 0\}$, and put $w(i,j) = a_{ij}$ for 
every $(i,j)\in E$. We use an index $i\in[n]$ to refer to both the $i$th row or column 
of $A$, and to the node $i$ of the digraph $G_A$. Thus, the non-zero entries of 
the $i$th column (the $i$th row, respectively) correspond to the arcs into (out of, 
respectively) node $i$. In the $L_1$ norm it is useful to think of the weight of an 
arc as a flow being carried by that arc. Thus, $\ali{i}$ is the total flow into vertex 
$i$ and $\alo{i}$ is the total flow out of it. Note that if a matrix $A$ is not balanced 
then for some nodes $i$, $\ali{i}\neq\alo{i}$, and thus the flow on the arcs does 
not constitute a valid circulation
because flow conservation is not maintained. Thus, the goal of balancing in the
$L_1$ norm can be stated as applying diagonal scaling to find a flow function on
the arcs of the graph $G_A$ that forms a valid circulation. We use both views
of the graph (with arc weights or flow), and also the matrix terminology, 
throughout this paper, as convenient.

Without loss of generality we may assume that the undirected graph underlying 
$G_A$ is connected. Otherwise, after permuting $V = \{1,\ldots, n\}$, the given 
matrix $A$ can be replaced by $\diag(A_1,\ldots, A_r)$ where each of $A_1,\ldots, A_r$ 
is a square matrix whose corresponding directed graph is connected. Thus, balancing 
$A$ is equivalent to balancing each of $A_1,\ldots, A_r$.

The goal of the iterative algorithm is to balance approximately a matrix $A$, up 
to an error term $\eps$. We define the error here.
\begin{definition}[approximate balancing]\label{def:approx}
Let $\eps > 0$.
\begin{enumerate}
\item A matrix $A$ is {\em $\eps$-balanced} iff\/
$\frac{\sqrt{\sum_{i=1}^n (\ali{i} - \alo{i})^2}}{\sum_{i,j} a_{i,j}} \le \eps$.

\item A diagonal matrix $D$ with positive diagonal entries is said to 
{\em $\eps$-balance} $A$ iff $DAD^{-1}$ 
is $\eps$-balanced.
\end{enumerate}
\end{definition}

\paragraph{The algorithms.}
Kalantari et al.~\cite{khachiyan} 
introduced the above definition of $\eps$-balancing, and showed that their 
algorithm for $\eps$-balancing a matrix in the $L_1$ norm uses 
$O(n^4 \ln((n/\eps)\ln w))$ arithmetic operations.
In their recent work, Schulman and Sinclair~\cite{schulman} use,
in the context of balancing in the $L_{\infty}$ norm, a
stronger notion of strict balancing (that requires even very low 
weight row-column pairs to be nearly balanced). Their iterative algorithm
strictly $\eps$-balances a matrix in the $L_{\infty}$ norm in 
$O(n^3\log (n\rho/\eps))$ iterations where $\rho$ measures the inital imbalance 
of the matrix. In this paper, we prove upper bounds on the convergence rate of 
the Osborne-Parlett-Reinsch (OPR) balancing. 

The OPR iterative algorithm balances indices in a fixed 
round-robin order. Schulman and Sinclair considered a variant that uses
a different rule to choose the next index to balance. We consider in this 
paper several alternative implementations of OPR balancing 
(including the original round-robin implementation) that differ only in the 
rule by which an index to balance is chosen at each step.
For all rules that we consider, the iteration generates 
a sequence $A = A^{(1)}, A^{(2)}, \ldots , A^{(t)}, \ldots$ of $n\times n$ of
matrices that converges to a unique balanced matrix $A^*$ (see Grad~\cite{grad} 
and Hartfiel~\cite{hartfiel}). The matrix $A^{(t+1)}$ is obtained by balancing an 
index of $A^{(t)}$. If the $i$th index of $A^{(t)}$ is chosen, we get that 
$A^{(t+1)} = D^{(t)} A^{(t)} {D^{(t)}}^{-1}$ where $D^{(t)}$ is a diagonal matrix 
with $d^{(t)}_{ii}= \sqrt{\|a_{.,i}^{(t)}\|_1/{\|a_{i,.}^{(t)}\|_1}}$ and $d^{(t)}_{jj} = 1$ 
for $j\neq i$. Note that $a_{i,.}^{(t)}$ ($a_{.,i}^{(t)}$, respectively) denotes the 
$i$th row ($i$th column, respectively) of $A^{(t)}$. Also, putting 
$\bar{D}^{(1)} = I_{n\times n}$ and $\bar{D}^{(t)} =  D^{(t-1)}\cdots D^{(1)}$ for 
$t>1$, we get that $A^{(t)} = \bar{D}^{(t)} A ({\bar{D}^{(t)}})^{-1}$.

The following lemma shows that each balancing step reduces the sum of entries 
of the matrix.
\begin{lemma}\label{lemma:reduction}
Balancing the $i$th index of a non-negative matrix $B=(b_{ij})_{n\times n}$ 
(with $b_{ii}=0$) decreases the total sum of the entries of $B$ by 
$(\sqrt{\|b_{.,i}\|_1}-\sqrt{\|b_{i,.}\|_1})^2$.
\end{lemma}

\begin{proof}
Before balancing, the total sum of entries in the $i$th row and in the $i$th 
column is $\|b_{i,.}\|_1 + \|b_{.,i}\|_1$. Balancing scales the entries of the 
$i$th column by $\sqrt{{\|b_{i,.}\|_1}/{\|b_{.,i}\|_1}}$ and entries of the $i$th 
row by $\sqrt{{\|b_{.,i}\|_1}/{\|b_{i,.}\|_1}}$. Thus, after balancing the sum of 
entries in the $i$th column, which equals the sum of entries in the $i$th row, 
is equal to $\sqrt{\|b_{i,.}\|_1\cdot\|b_{.,i}\|_1}$. The entries that are not in 
the balanced row and column are not changed. Therefore, keeping in mind 
that $b_{ii}=0$, balancing decreases $\sij b_{ij}$ by $\|b_{.,i}\|_1 + 
\|b_{i,.}\|_1- 2\sqrt{\|b_{i,.}\|_1\cdot\|b_{.,i}\|_1} = 
(\sqrt{\|b_{.,i}\|_1}-\sqrt{\|b_{i,.}\|_1})^2$.
\end{proof}

\paragraph{A reduction to convex optimization.}
Kalantari et al.~\cite{khachiyan}, as part of their algorithm, reduce matrix
balancing to a convex optimization problem. We overview their reduction
here. Our starting point is Osborne's observation that if a diagonal matrix 
$D=\diag(d_1,\ldots,d_n)$ balances a matrix $A$ in the $L_2$ norm, then 
the diagonal vector $\mathbf{d} =(d_1, \ldots, d_n)$ minimizes the Frobenius 
norm of the matrix $DAD^{-1}$. The analogous claim for the $L_1$ norm 
is that if a diagonal matrix $D=\diag(d_1,\ldots,d_n)$ balances a matrix $A$ 
in the $L_1$ norm, then the diagonal vector $\mathbf{d} =(d_1, \ldots, d_n)$ 
minimizes the function $F(\mathbf{d}) = \sij\aij\frac{d_i}{d_j}$. On the other 
hand, Eaves et al.~\cite{eaves} observed that a matrix $A$ can be balanced 
if and only if the digraph $G_A$ is strongly connected. The following theorem~\cite[Theorem~1]{khachiyan} summarizes the above discussion.
\begin{theorem}[Kalantari et al.]\label{thm:balancable}
Let $A=(a_{ij})_{n\times n}$ be a real non-negative matrix, $a_{ii} = 0$, for 
all $i=1,\ldots n$, such that the undirected graph underlying $G_A$ is connected. 
Then, the following statements are equivalent.
\begin{enumerate}[(i)]
\item $A$ is balanceable (i.e., there exists a diagonal matrix $D$ such that 
$DAD^{-1}$ is balanced).
\item $G_A$ is strongly connected.
\item Let $F(\mathbf{d}) = \sum_{(i,j)\in E} a_{ij}\frac{d_i}{d_j}$. There is a point
$\mathbf{d}^*\in\Omega=\{d\in\mathbb{R}^n:\ d_i >0, i = 1,\ldots,n\}$ such that
$F(\mathbf{d}^*) = \inf\{F(\mathbf{d}):\ \mathbf{d}\in\Omega\}$.
\end{enumerate}
\end{theorem}
We refer the reader to~\cite[Theorem~1]{khachiyan} for a proof. 
We have the following corollary.
\begin{corollary}\label{cor: minimizer balances}
$\mathbf{d}^*$ minimizes $F$ over $\Omega$ if and only if 
$D^*=\diag(d^*_1,\ldots,d^*_n)$ balances $A$.
\end{corollary}

\begin{proof}
As $F$ attains its infimum at $\mathbf{d}^*\in\Omega$,
its gradient $\nabla F$ satisfies $\nabla F(\mathbf{d}^*) = 0$.
Also,
$\frac{\partial{F}(\mathbf{d}^*)}{\partial{d_i}}  = 0$ if and only if 
$\sum_{j=1}^n a_{ij}\cdot(d^*_i/d^*_j) = \sum_{j=1}^n a_{ji}\cdot(d^*_j/d^*_i)$
for all $i\in[n]$. In other words, $\nabla F(\mathbf{d}^*) = 0$ if and 
only if the matrix $D^*A{D^*}^{-1}$ is balanced where 
$D^*=\diag(d^*_1,\ldots,d^*_n)$. Thus, $\mathbf{d}^*$ minimizes $F$ 
over $\Omega$ if and only if $D^*=\diag(d^*_1,\ldots,d^*_n)$ balances $A$.
\end{proof}
 
It can also be shown that under the assumption of Theorem~\ref{thm:balancable}, 
the balancing matrix $D^*$ is unique up to a scalar factor 
(see Osborne~\cite{osborne} and Eaves et al.~\cite{eaves}). Therefore, the 
problem of balancing matrix $A$ can be reduced to optimizing the function $F$. 
Since we are optimizing over the set $\Omega$ of strictly positive vectors, we 
can apply a change of variables $\dd=(e^{x_1},\ldots,e^{x_n})\in\mathbb{R}^n$ 
to obtain a convex objective function:
\begin{equation}\label{def:F}
f(\x) = f_A(\x) = \sum_{i,j=1}^n  a_{ij}e^{x_i - x_j}.
\end{equation}
Kalantari et al.~\cite{khachiyan} use the convex function $f$ because it can be 
minimized using the ellipsoid algorithm. We do not need the convexity of $f$, 
and use $f$ instead of $F$ only because it is more convenient to work with,
and it adds some intuition. 
Notice that the partial derivative of $f$ with respect to $x_i$ is
\begin{equation}\label{eq:partial}
\frac{\partial{f(\x)}}{\partial{x_i}} = \sum_{j=1}^n a_{ij}\cdot e^{x_i - x_j} - 
\sum_{j=1}^n a_{ji}\cdot e^{x_j - x_i},
\end{equation}
which is precisely the difference between the $L_1$ norms of the $i$th row and 
the $i$th column of the matrix $DAD^{-1}$, where $D=\diag(e^{x_1}, \ldots ,e^{x_n})$. 
Also, by definition, the diagonal matrix 
$\diag(e^{x_1},\ldots,e^{x_n})$ $\eps$-balances $A$ iff
\begin{equation}\label{eq:stop condition}
\frac{\|\nabla f(\x)\|_2}{f(\x)} =
\frac{\sqrt{\sum_{i=1}^n \left(\sum_{j=1}^n a_{ij}e^{x_i - x_j} - 
\sum_{j=1}^n a_{ji}e^{x_j - x_i}\right)^2}}{\sum_{i,j=1}^n  a_{ij}e^{x_i - x_j}}\le \eps.
\end{equation}

We now state and prove a key lemma that our analysis uses.
The lemma uses combinatorial flow and circulation arguments 
to measure progress 
by bounding $f(\x)-f(\x^*)$ in terms of $\|\nabla f(\x)\|_1$ which is a 
global measure of imbalances of all vertices. 
\begin{lemma}\label{lemma:opt lower}
Let $f$ be the function defined in Equation~\eqref{def:F}, and let $\x^*$ 
be a global minimum of $f$. Then, for all $\x\in\mathbb{R}^n$, 
$f(\x)-f(\x^*) \le \frac{n}{2}\cdot\|\nabla f(\x)\|_1$.
\end{lemma}

\begin{proof}
Recall that $f(\x) = f_A(\x)$ is the sum of entries of a matrix 
$B = (b_{ij})$ defined by $b_{ij} = a_{ij}\cdot e^{x_i-x_j}$.
Notice that 
$f(\x) = f_B(\vec{\mathbf{0}})$, and $f(\x^*) = f_B(\x^{**})$,
where $\x^{**} = \x^* - \x$.  
Alternatively, $f(\x)$ is the sum of flows (or weights) of the arcs of 
$G_B$, and $f(\x^*)$ is the sum of flows of the arcs of 
a graph $G^*$ (an arc $ij$ of $G^*$ carries a flow of 
$a_{ij}\cdot e^{x^*_i-x^*_j}$). 
Notice that $G_B$ and $G^*$ have the same set 
of arcs, but with different weights. By Equation~\eqref{eq:partial},
$\|\nabla f_A(\x)\|_1 = \sum_{i=1}^n \big|\|b_{.,i}\|_1-\|b_{i,.}\|_1\big|$, 
i.e., it is the sum over all the nodes of $G_B$ of the difference 
between the flow into the node and flow out of it. Also notice that 
$G_B$ is unbalanced (else the statement of the lemma is trivial), 
however $G^*$ is balanced. Therefore, the arc flows in $G^*$, but
not those in $G_B$, form a valid circulation. 

Our proof now proceeds in two main steps. In the first step we show 
a way of reducing the flow on some arcs of $G_B$, such that the 
revised flows make every node balanced (and thus form a valid
circulation). We also make sure that 
the total flow reduction is at most $\frac{n}{2}\cdot\|\nabla f_A(\x)\|_1$. 
In the second step we show that sum of revised flows of all the arcs 
is a lower bound on $f(\x^*)$. These two steps together prove the lemma. 

We start with the first step. The nodes of $G_B$ are not balanced. Let 
$S$ and $T$ be a partition of the unbalanced nodes of $G_B$, with 
$S=\left\{i\in[n]:\ \|b_{.,i}\|_1 > \|b_{i,.}\|_1\right\}$ and 
$T=\left\{i\in[n]:\ \|b_{.,i}\|_1 < \|b_{i,.}\|_1\right\}$. That is, 
the flow into a node in $S$ exceeds the flow out of it, and the
flow into a node in $T$ is less than the flow out of it. We have
that
$$
\sum_{i\in S}(\|b_{.,i}\|_1-\|b_{i,.}\|_1) - \sum_{i\in T}(\|b_{i,.}\|_1-\|b_{.,i}\|_1) = 
\sum_{i\in[n]}(\|b_{.,i}\|_1-\|b_{i,.}\|_1)  = 0.
$$
Thus, we can view each node $i\in S$ as a source with supply 
$\|b_{.,i}\|_1-\|b_{i,.}\|_1$, and each node $i\in T$ as a sink with 
demand $\|b_{i,.}\|_1-\|b_{.,i}\|_1$, and the total supply equals 
the total demand. We now add some weighted arcs connecting 
the nodes in $S$ to the nodes in $T$. These arcs carry the supply
at the nodes in $S$ to the demand at the nodes in $T$. Note that 
we may add arcs that are parallel to some existing arcs in $G_B$. 
Such arcs can be replaced by adding flow to the parallel existing 
arcs of $G_B$. In more detail, to compute the flows of the added 
arcs (or the added flow to existing arcs), we add arcs inductively 
as follows. We start with any pair of nodes $i\in S$ and $j\in T$, and 
add an arc from $i$ to $j$ carrying flow equal to the minimum between 
the supply at $i$ and the demand at $j$. Adding this arc will balance 
one of its endpoints, but in the new graph the sum of supplies at the
nodes of $S$ is still equal to the sum of demands at the nodes of $T$, 
so we can repeat the process. (Notice that either $S$ or $T$ or both
lose one node.) Each additional arc balances at least one unbalanced 
node, so $G_B$ gets balanced by adding at most $n$ additional arcs 
from nodes in $S$ to nodes in $T$. The total flow on the added arcs 
is exactly $\sum_{i\in S}(\|b_{.,i}\|_1-\|b_{i,.}\|_1) = \frac{1}{2}\cdot\|\nabla f(\x)\|_1$. 

Let $E'$ be the set of newly added arcs, and let $G_{B'}$ be the new graph
with arc weights given by $B' = (b'_{ij})$. 
Since $G_{B'}$ is balanced, the arc flows form a valid circulation. We next 
decompose the total flow of arcs into cycles. Consider a cycle $C$ in 
$G_{B'}$ that contains at least one arc from $E'$ (i.e., $C\cap E'\ne\emptyset$).
Reduce the flow on all arcs in $C$ by $\alpha = \min_{ij\in C} b'_{ij}$. 
This can be viewed as peeling off from $G_{B'}$ a circulation carrying 
flow $\alpha$. This reduces the flow on at least one arc to zero, and 
the remaining flow on arcs is still a valid circulation, so we
can repeat the process. It can be repeated as long as there is positive 
flow on some arc in $E'$. Eliminating the flow on all arcs in $E'$ 
using cycles reduces the total flow on the arcs by at most $n$ times 
the total initial flow on the arcs in $E'$ (i.e., $\frac{n}{2}\cdot\|\nabla f(\xt{1})\|_1$), 
because each cycle contains at most $n$ arcs and its flow $\alpha$
that is peeled off reduces the flow on at least one arc in $E'$ by $\alpha$. 
After peeling off all the flow on all arcs in $E'$, all the arcs with positive 
flow are original arcs of $G_B$. Let $G_{B''}$ be the graph with the 
remaining arcs
and their flows which are given by $B'' = (b''_{ij})$. 
The total flow on the arcs 
of $G_{B''}$ is at least 
$f(\x) + \frac{1}{2}\cdot\|\nabla f(\x)\|_1 -\frac{n}{2}\cdot\|\nabla f(\x)\|_1 \ge 
f(\x) - \frac{n}{2}\cdot\|\nabla f(\x)\|_1$.

Next we show that the total flow on the arcs of $G_{B''}$ is a lower 
bound on $f(\x^*)$. Our key tool for this is the fact that balancing operations 
preserve the product of arc flows on any cycle in the original graph 
$G_B$, because balancing a node $i$ multiplies the flow on the arcs 
into $i$ by some factor $r$ and the flow on the arcs out of $i$ by $\frac{1}{r}$. 
Thus, the geometric mean of the flows of the arcs on any cycle is not 
changed by a balancing operation.
The arc flows in $G_{B''}$ form a valid circulation, and thus can 
be decomposed 
into flow cycles $C_1,\ldots,C_q$ by a similar peeling-off process 
that was described 
earlier. Let $n_1,\ldots, n_q$ be the lengths of cycles, and let 
$\alpha_1,\ldots, \alpha_q$ 
be their flows. The total flow on arcs in $G_{B''}$ is, therefore, 
$\sum_{k=1}^q n_k\cdot\alpha_k$. Notice that, by construction, 
$b''_{ij} \le b_{ij}$, and the decomposition into cycles gives that
$b''_{ij} = \sum_{k:ij\in C_k}\alpha_k$. Thus,
$f(\x^*)= \sum_{i,j=1}^n b_{ij}e^{x^{**}_i-x^{**}_j} \ge 
\sum_{i,j=1}^n b''_{ij} e^{x^{**}_i-x^{**}_j} = 
\sum_{i,j=1}^n\sum_{k:ij\in C_k}\alpha_k e^{x^{**}_i-x^{**}_j}
= \sum_{k=1}^q \sum_{ij\in C_k}\alpha_k e^{x^{**}_i-x^{**}_j} \ge  
\sum_{k=1}^q n_k 
\left(\prod_{ij\in C_k} \alpha_k e^{x^{**}_i-x^{**}_j}\right)^{1/n_k}
= \sum_{k=1}^q n_k\alpha_k = \sum_{i,j = 1}^n b''_{ij}$,
where the last inequality uses the arithmetic-geometric mean inequality.
Notice that the right-hand side is the total flow on the arcs of $G_{B''}$, 
which is at least $f(\x) - \frac{n}{2}\cdot\|\nabla f(\xt{1})\|_1$. Thus, 
$f(\x^*)\ge f(\x) - \frac{n}{2}\cdot\|\nabla f(\x)\|_1$, and this completes 
the proof of the lemma.
\end{proof}

\section{Greedy Balancing}\label{sec:greedy}

Here we present and analyze a greedy variant of the OPR iteration.
Instead of balancing indices in a fixed round-robin order, the greedy 
modification chooses at iteration $t$ an index $i_t$ of $A^{(t)}$ such 
that balancing the chosen index results in the largest decrease in the 
sum of entries of $A^{(t)}$. In other words, we pick $i_t$ such that 
the following equation holds.
\begin{equation}\label{eq:greedy}
i_t = \argmax_{i\in[n]}{\left(\sqrt{\|a^{(t)}_{.,i}\|_1}-\sqrt{\|a^{(t)}_{i,.}\|_1}\right)^2}
\end{equation}

We give two analyses of this variant, one that shows that the
number of operations is nearly linear in the size of $G_A$,
and another that shows that the number of operations is
nearly linear in $1/\eps$. More specifically, we prove the
following theorem.
\begin{theorem}\label{thm:greedy}
Given an $n\times n$ matrix $A$, let $m = |E(G_A)|$,
the greedy implementation of the OPR iterative algorithm 
outputs an $\eps$-balanced matrix in $K$ iterations 
which cost a total of $O(m + Kn\log n)$ arithmetic operations
over $O(n\log w)$-bit numbers,
where 
$K = O\left(\min\left\{\eps^{-2}\log w,\eps^{-1}n^{3/2}\log(w/\eps)\right\}\right)$.
\end{theorem}

The proof uses the convex optimization framework introduced in 
Section~\ref{sec:preliminaries}. Recall that 
$A^{(t)} = \bar{D}^{(t)} A ({\bar{D}^{(t)}})^{-1}$. If we let 
$\bar{D}^{(t)}=\diag(e^{x_1^{(t)}}, \ldots, e^{x_n^{(t)}})$, the iterative 
sequence can be viewed as generating a sequence of points 
$\xt{1},\xt{2}, \ldots\allowbreak,\xt{t}, \ldots$ in $\mathbb{R}^n$,  
where $\xt{t} = (x_1^{(t)}, \ldots, x_n^{(t)})$ and 
$A^{(t)} = \bar{D}^{(t)} A ({\bar{D}^{(t)}})^{-1}=(a_{ij}e^{x_i^{(t)}-x_j^{(t)}})_{n\times n}$. 
Initially, $\xt{1}=(0,\ldots,0)$, and $\xt{t+1} = \xt{t} + \alpha_t\mathbf{e}_i$, 
where $\alpha_t = \ln(d^{(t)}_{ii})$ and $\mathbf{e}_i$ is the $i$th vector of 
the standard basis for $\mathbb{R}^n$. By Equation~\eqref{def:F}, the value 
$f(\xt{t})$ is sum of the entries of the matrix $A^{(t)}$. 
The following key lemma allows us to lower bound the decrease in the value 
of $f(\xt{t})$ in terms of a value that can be later related to the stopping condition.
\begin{lemma}\label{lemma:progress}
If index $i_t$ defined in Equation~\eqref{eq:greedy} is picked to 
balance $A^{(t)}$, then 
$f(\xt{t})-f(\xt{t+1)}) \ge \frac{\|\nabla f(\xt{t})\|_2^2}{4f(\xt{t})}$.
\end{lemma}

\begin{corollary}\label{cor:progress}
If matrix $A^{(t)}$ is not $\eps$-balanced, by balancing index $i_t$ at 
iteration $t$, we have 
$f(\xt{t})-f(\xt{t+1)}) \ge \frac{\eps^2}{4}\cdot f(\xt{t})$.
\end{corollary}

\begin{proof}[Proof of Theorem~\ref{thm:greedy}]
By Corollary~\ref{cor:progress}, while $A^{(t)}$ is not 
$\eps$-balanced, there exists an index $i_t$ to balance such 
that $f(\xt{t})-f(\xt{t+1)}) \ge \frac{\eps^2}{4}\cdot f(\xt{t})$.
Thus, $f(\xt{t+1})\le \left(1-\displaystyle\frac{\eps^2}{4}\right)\cdot f(\xt{t})$. 
Iterating for $t$ steps yields
$f(\xt{t+1})\le \left(1-\frac{\eps^2}{4}\right)^{t}\cdot f(\xt{1})$.
So, on the one hand, $f(\xt{1})=\sum_{i,j=1}^n a_{ij}$ since $f(\xt{1})$ 
is the sum of entries in $A^{(1)}$. On the other hand, we argue that the 
value of $f(\xt{t+1})$ is at least $\min_{(i,j) \in E}a_{ij}$. To see this, 
consider a directed cycle in the graph $G_A$. It's easy to see that 
balancing operations preserve the product of weights of the arcs on 
any cycle. Thus, the weight of at least one arc in the cycle is at least 
its weight in the input matrix $A$. Therefore,
$a_{\min}\le f(\xt{t+1})\le \left(1-\frac{\eps^2}{4}\right)^{t}\cdot f(\xt{1})=\left(1-\frac{\eps^2}{4}\right)^{t}\cdot \sum_{i,j=1}^n a_{ij}$.
Thus, 
$t\le\frac{4}{\eps^2}\cdot \ln w$ 
and this is an upper bound on the number of balancing operations 
before an $\eps$-balanced matrix is obtained. 
The algorithm initially computes $\ali{i}$ and $\alo{i}$ for all $i\in[n]$ in 
$O(m)$ time. Also the algorithm initially computes the value of $\left(\sqrt{\alo{i}}-\sqrt{\ali{i}}\right)^2$ for all $i$ in $O(m)$ time 
and inserts the values in a priority queue in $O(n\log n)$ time. 
The values of $\|a_{i,.}^{(t)}\|_1$, $\|a_{.,i}^{(t)}\|_1$ for all $i$ 
and $\left(\sqrt{\|a_{i,.}^{(t)}\|_1} - \sqrt{\|a_{.,i}^{(t)}\|_1}\right)^2$ 
are updated after each balancing operation. In each iteration the 
weights of at most $n$ arcs change. Updating the values of 
$\|a_{i,.}^{(t)}\|_1$ and $\|a_{.,i}^{(t)}\|_1$ takes $O(n)$ time and 
updating the values of 
$\left(\sqrt{\|a_{i,.}^{(t)}\|_1} - \sqrt{\|a_{.,i}^{(t)}\|_1}\right)^2$ 
involves at most $n$ updates of values in the priority queue, each 
taking time $O(\log n)$. Thus, the first iteration takes $O(m)$ 
operations and each iteration after that takes $O(n\log n)$ operations, 
so the total running time of the algorithm in terms of arithmetic operations 
is $O(m + (n\log n\log w)/\eps^2)$. 

An alternative analysis completes the proof.
Notice that $\|\nabla f(\xt{t})\|_2\le \|\nabla f(\xt{t})\|_1 \le \sqrt{n}\cdot \|\nabla f(\xt{t})\|_2$.
Therefore,
$f(\xt{t})-f(\xt{t+1)}) \ge \frac{\|\nabla f(\xt{t})\|_2^2}{4f(\xt{t})}
\ge\frac{\|\nabla f(\xt{t})\|_2}{4\sqrt{n}\cdot f(\xt{t})}\cdot \|\nabla f(\xt{t})\|_1
\ge \frac{1}{2n^{3/2}}\cdot\frac{\|\nabla f(\xt{t})\|_2}{f(\xt{t})}\cdot(f(\xt{t}) - f(\x^*))$,
where the first inequality follows from Lemma~\ref{lemma:progress},
and the last inequality follows from Lemma~\ref{lemma:opt lower}.
Therefore, while $A^t$ is not $\eps$-balanced (so
$\frac{\|\nabla f(\xt{t})\|_2}{f(\xt{t})} > \eps$), we have that
$f(\xt{t})-f(\xt{t+1)}\ge \frac{\eps}{2n^{3/2}}\cdot (f(\xt{t})-f(\x^*))$.
Rearranging the terms, we get
$f(\xt{t+1})-f(\x^*)\le \left(1-\frac{\eps}{2n^{3/2}}\right)\cdot (f(\xt{t}) - f(\x^*))$.
Therefore,
$f(\xt{t+1})-f(\x^*)\le \left(1-\frac{\eps}{2n^{3/2}}\right)^t\cdot (f(\xt{1}) - f(\x^*))$.
Notice that by Lemma~\ref{lemma:progress},
$f(\xt{t+1})-f(\x^*)\ge f(\xt{t+1}) - f(\xt{t+2})\ge 
\left(\frac{\|\nabla f(\xt{t+1})\|_2}{2f(\xt{t+1})}\right)^2\cdot f(\xt{t+1}) \ge
\left(\frac{\|\nabla f(\xt{t+1})\|_2}{2f(\xt{t+1})}\right)^2\cdot a_{\min}$.
On the other hand, $f(\xt{1}) - f(\x^*)\le f(\xt{1})\le \sum_{i,j=1}^n a_{ij}$.
Thus, for  
$t=2\eps^{-1}\cdot n^{3/2}\ln(4w/\eps^2)$,
we have that $\frac{\|\nabla f(\xt{t+1})\|_2}{f(\xt{t+1})}\le\eps$, so the matrix
is $\eps$-balanced.
\end{proof}

\section{Round-Robin Balancing (the original algorithm)}\label{original}
Recall that original Osborne-Parlett-Reinsch algorithm balances indices 
in a fixed round-robin order. Although the greedy variant of the OPR 
iteration is a simple modification of the implementation, the convergence
rate of the original algorithm (with no change) is interesting. This is 
important because the original algorithm has a slightly simpler 
implementation, and also because this is the implementation used in 
almost all numerical linear algebra software including MATLAB, LAPACK 
and EISPACK (refer to ~\cite{spectra2005, kressner} for further background). 
We give some answer to this question in the following theorem.
\begin{theorem}\label{thm:original}
Given an $n\times n$ matrix $A$, the original implementation of the 
OPR iteration outputs an $\eps$-balanced matrix in $O(\eps^{-2}n^2\log w)$ 
iterations totalling $O(\eps^{-2}mn\log w)$ arithmetic operations
over $O(n\log w)$-bit numbers
($m$ is the number of non-zero entries of $A$).
\end{theorem}

\section{Randomized Balancing}\label{randomized}

In Theorem~\ref{thm:greedy} the arithmetic operations were applied to $O(n\ln w)$-bit
numbers. This will cause an additional factor of $O(n\ln w)$ in the running time of the
algorithm. In this section we fix this issue by presenting a randomized variant of the 
algorithm that applies arithmetic operations to numbers of $O(\ln(wn/\eps))$ bits. Thus, 
we obtain a algorithm for balancing that runs in nearly linear time. While the greedy 
algorithm works by picking the node $i$ that maximizes 
$(\sqrt{\|a_{i,.}\|}-\sqrt{\|a_{.,i}\|})^2$, the key idea of the randomized algorithm is 
sampling a node for balancing using sampling probabilities that do not depend on 
the difference in arc weights (the algorithm uses low-precision rounded weights, so 
this can affect significantly the difference). Instead, our sampling probabilities depend 
on the sum of weights of the arcs incident on a node.

We first introduce some notation. We use $O(\ln(wn/\eps))$ bits of precision to 
approximate $x_i$-s with $\hx_i$-s. 
Thus, $x_i - 2^{-O(\ln(nw/\eps))} \le \hx_i \le x_i$. In addition to maintaining 
$\hxt{t} = (\hx_1^{(t)}, \hx_2^{(t)}, \ldots, \hx_n^{(t)})$ at every time $t$. The 
algorithm also maintains for every $i$ and $j$ the value of $\haij$ which is 
$a_{ij}^{(t)} = a_{ij}e^{\hx_i^{(t)}- \hx_j^{(t)}}$ truncated to $O(\ln(wn/\eps))$ 
bits of precision. We set the hidden constant to give a truncation error of
$r = (\eps/wn)^{10} a_{\min}$, so $a_{ij}^{(t)} - r \le \haij \le a_{ij}^{(t)}$. The 
algorithm also maintains for every $i$, $\hao = \sum_{j=1}^n \haij$ and 
$\hai = \sum_{j=1}^n\haji$. For every $i$, we use the notation 
$\|a_{i,.}^{(t)}\| = \sum_{j=1}^n a_{ij}^{(t)}$ and 
$\|a_{.,i}\| = \sum_{j=1}^n a_{ji}^{(t)}$. Note that the algorithm does not 
maintain the values $a_{ij}^{(t)}$, $\|a_{.,i}^{(t)}\|$ or $\|a_{i,.}^{(t)}\|$. 

The algorithm works as follows (see the pseudo-code of Algorithm~\ref{alg:random}
that appears in Section~\ref{proofs}). 
In each iteration it samples an index $i$ with probability 
$p_i = \frac{\hao + \hai}{2\sum_{i,j}\haij}$. If $i$ is sampled, a balancing operation 
is applied to index $i$ only if the arcs incident on $i$ have significant weight,
and $i$'s imbalance is sufficiently large. Put $\hat{M}_i = \max\{\hao, \hai\}$ 
and put $\hat{m}_i = \min\{\hao, \hai\}$.
The imbalance is considered large 
if $\hat{m}_i=0$ (this can happen because of the low precision), or 
if $\hat{m}_i\neq 0$ and 
$\frac{\hat{M}_i}{\hat{m}_i} \ge 1 + \frac{\eps}{n}$.
A balancing operation is done by adding $\alpha$ to $\hxt{t}_i$,
where $\alpha = \frac{1}{2}\ln(\hai/\hao)$, unless $\hat{m}_i=0$,
in which case we replace the $0$ value by $nr$.
This updates the weights of the arcs incident on $i$. Also, the $L_1$ norms 
of changed rows and columns are updated. (For convenience we use 
in this section $\|\cdot\|$ instead of $\|\cdot\|_1$ to denote the $L_1$ norm.)

Note that in the pseudo-code, $\gets$ indicates an assignment where
the value on the right-hand side is computed to $O(\ln(wn/\eps))$ bits 
of precision. Thus, we have
\begin{equation}\label{eq28}
\alpha - (\eps/wn)^{10} \le \widehat{x}^{(t+1)}_i - \widehat{x}^{(t)}_i \le \alpha
\end{equation}
and
$$
a_{ij}e^{\hx_i^{(t+1)}- \hx_j^{(t+1)}}-r\le\widehat{a}_{ij}^{(t+1)}\le a_{ij}e^{\hx_i^{(t+1)}- \hx_j^{(t+1)}},
$$
$$
a_{ji}e^{\hx_i^{(t+1)}- \hx_j^{(t+1)}}-r\le\widehat{a}_{ji}^{(t+1)}\le a_{ji}e^{\hx_i^{(t+1)}- \hx_j^{(t+1)}}.
$$
\begin{theorem}\label{thm:random}
With probability at least $\frac{9}{10}$, Algorithm~\ref{alg:random} returns 
in time $O(m\ln\sum_{ij}a_{ij} + \eps^{-2} n\ln(wn/\eps)\ln w)$
an $\eps$-balanced matrix.
\end{theorem}

The idea of proof is to show that in every iteration of the algorithm we
reduce $f(.)$ by at least a factor of $1-\Omega(\eps^2)$. Before we
prove the theorem, we state and prove a couple of useful lemmas.

Fix an iteration $t$, and define three sets of indices as follows:
$A=\{i:\hao + \hai \ge \eps a_{\min}/10wn\}$, 
$B = \{i: \hat{m}_i \ne 0\wedge \hat{M}_i/\hat{m}_i \ge 1+ {\eps}/{n}\}$, 
and $C=\{i: \hat{m}_i = 0\}$. If the random index $i$ satisfies 
$i\notin A$ or $i\in A\setminus({B\cup C})$, the algorithm does not 
perform any balancing operation on $i$. The following lemma states 
that the expected decrease due to balancing such indices is small, 
and thus skipping them does not affect the speed of convergence
substantially.
\begin{lemma}\label{lemma:case}
For every iteration $t$,
$\sum_{i\notin A\cap (B\cup C)} p_i\cdot
\left(\sqrt{\|a_{i,.}^{(t)}\|}-\sqrt{\|a_{.,i}^{(t)}\|}\right)^2 < 
\frac{2\eps^2}{n}\cdot f(\xt{t})$, where $p$ is the probability
distribution over indices at time $t$.
\end{lemma}

We now show a lower bound on the decrease in $f(\cdot)$, if a node 
$i\in A\cap(B\cup C)$ is balanced.
\begin{lemma}\label{lemma:decrease}
If $i\in A\cap (B \cup C)$ is balanced in iteration $t$, then 
$f(\hxt{t+1}) - f(\hxt{t}) \ge 
\frac{1}{10}\cdot\left(\sqrt{\|a_{i,.}^{(t)}\|}-\sqrt{\|a_{.,i}^{(t)}\|}\right)^2$.
\end{lemma}

\begin{proof}[Proof of Theorem~\ref{thm:random}]
By Lemma~\ref{lemma:decrease}, the expected decrease in 
$f(.)$ in iteration $t$ is lower bounded as follows.
\begin{eqnarray*}
                        \mathbb{E}[f(\hxt{t}) - f(\hxt{t+1})] 
&\ge& \sum_{i\in A\cap(B\cup C)} p_i\cdot \frac{1}{10} \left(\sqrt{\|a_{i,.}^{(t)}\|}-
                \sqrt{\|a_{.,i}^{(t)}\|}\right)^2 \\
& = & \frac{1}{10}\cdot \left(\sum_{i = 1}^n p_i\cdot 
        \left(\sqrt{\|a_{i,.}^{(t)}\|}-\sqrt{\|a_{.,i}^{(t)}\|}\right)^2  - 
        \sum_{i\notin A\cap (B\cup C)} p_i\cdot\left(\sqrt{\|a_{i,.}^{(t)}\|}-
          \sqrt{\|a_{.,i}^{(t)}\|}\right)^2 \right).
\end{eqnarray*}
The second term can be bounded, using Lemma~\ref{lemma:case}, by
$\sum_{i\notin A\cap (B\cup C)}\left(\sqrt{\|a_{i,.}^{(t)}\|}-
\sqrt{\|a_{.,i}^{(t)}\|}\right)^2 \le \frac{2\eps^2}{n}\cdot f(\hxt{t})$.
For the first term, we can write 
\begin{eqnarray*}
                 \sum_{i = 1}^n p_i\cdot \left(\sqrt{\|a_{i,.}^{(t)}\|}-\sqrt{\|a_{.,i}^{(t)}\|}\right)^2 
&\ge& \sum_{i = 1}^n p_i \cdot\frac{(\|a_{i,.}^{(t)}\| - \|a_{.,i}^{(t)}\|)^2}
                            {2(\|a_{i,.}^{(t)}\| + \| a_{.,i}^{(t)}\|)} \\
& = &  \sum_{i=1}^n \frac{\hao + \hai}{2\sum_{ij}\haij} \cdot 
             \frac{(\|a_{i,.}^{(t)}\| - \|a_{.,i}^{(t)}\|)^2}{2(\|a_{i,.}^{(t)}\| + \|a_{.,i}^{(t)}\|)} \\
&\ge& \frac{1}{16}\cdot\sum_{i=1}^n 
                  \frac{(\|a_{i,.}^{(t)}\| - \|a_{.,i}^{(t)}\|)^2}{\sum_{ij}a_{ij}^{(t)}} \\
& = &  \frac{\|\nabla f(\hxt{t})\|_2^2}{16f(\hxt{t})}\ge 
           \frac{\eps^2}{16}\cdot f(\hxt{t}),
\end{eqnarray*}
where the penultimate inequality holds because $\frac{\hat{M}_i}{M_i} \ge \frac{1}{2}$,
so $\frac{\hao + \hai}{\|a_{i,.}^{(t)}\| + \|a_{.,i}^{(t)}\|} \ge \frac{\hat{M}_i}{2M_i} \ge \frac{1}{4}$,
and the last inequality holds as long as the matrix is not $\eps$-balanced,
so $\frac{\|\nabla f(\hxt{t})\|_2}{f(\hxt{t})}\ge \eps$. 
Combining everything together, we get
\begin{eqnarray*}
                       \mathbb{E}[f(\hxt{t}) - f(\hxt{t+1})]
&\ge& \frac{1}{10}\cdot \left(\sum_{i = 1}^n p_i\cdot 
                 \left(\sqrt{\|a_{i,.}^{(t)}\|}-\sqrt{\|a_{.,i}^{(t)}\|}\right)^2  - 
                 \sum_{i\notin A\cap (B\cup C)}p_i\cdot
                 \left(\sqrt{\|a_{i,.}^{(t)}\|}-\sqrt{\|a_{.,i}^{(t)}\|}\right)^2 \right) \\
&\ge&  \frac{1}{10}\cdot \left(\frac{\eps^2}{16}\cdot f(\hxt{t}) - 
                \frac{2\eps^2}{n}\cdot f(\hxt{(t)})\right) \ge \frac{\eps^2}{320}\cdot f(\hxt{t}),
\end{eqnarray*}
where the last inequality assumes $n \ge 64$. 
This implies that the expected number of iterations to obtain an 
$\eps$-balanced matrix is $O(\eps^{-2}\ln w)$. Markov's inequality
implies that with probability $\frac{9}{10}$ an $\eps$-balanced matrix 
is obtained in $O(\eps^{-2}\ln w)$ iterations. It is easy to see that each 
iteration of the algorithm takes $O(n\ln (wn/\eps))$ time. Initializations 
take $O(m\ln\sum_{ij}a_{ij})$ time. So the total running time of the algorithm 
is $O(m\ln\sum_{ij}a_{ij} + \eps^{-2} n\ln(wn/\eps)\ln w)$.
\end{proof}



\section{A Lower Bound on the Rate of Convergence}\label{lower bound}


In this section we prove the following lower bound.
\begin{theorem}\label{thm: lower main}
There are matrices for which
all variants of the Osborne-Parlett-Reinsch iteration
(i.e., regardless of the order of indices chosen to balance)
require $\Omega(1/\sqrt{\eps})$ iterations to balance
the matrix to the relative error of $\eps$.
\end{theorem}

Before proving this theorem, we present the claimed construction.
Let $A$ be the following $4\times 4$ matrix, and let $A^*$ denote
the corresponding fully-balanced matrix.
\begin{equation*}
A = 
\begin{bmatrix}
0 & 1 & 0 & 0 \\ 
1 & 0 & \beta+\eps & 0 \\ 
0 & \eps & 0 & 1\\
0 & 0 & 1& 0
 \end{bmatrix}\text{~~,~~}
A^*=
\begin{bmatrix}
0 & 1 & 0 & 0 \\ 
1 & 0 & \sqrt{\eps(\beta+\eps)} & 0 \\ 
0 & \sqrt{\eps(\beta+\eps)} & 0 & 1\\
0 & 0 & 1& 0
 \end{bmatrix}
\end{equation*}
Here $\eps > 0$ is arbitrarily small, and $\beta = 100\eps$. 
It's easy to see that $A^* = D^* A {D^*}^{-1}$ where 
$$
D=\allowbreak\diag\left(1,1,\sqrt{\frac{\beta + \eps}{\eps}}, \sqrt{\frac{\beta + \eps}{\eps}}\right). 
$$
To prove Theorem~\ref{thm: lower main}, we show that balancing $A$ to the relative error of $\eps$
requires $\Omega(1/\sqrt{\eps})$ iterations, regardless of the order of balancing 
operations. Notice that in order to fully balance $A$, we simply need to replace
$a_{23}$ and $a_{32}$ by their geometric mean. We measure the rate of 
convergence using the ratio $a_{32}/a_{23}$. This ratio is initially
$\frac{\eps}{\beta+\eps} = \frac{1}{101}$. When the matrix is fully balanced, 
the ratio becomes $1$. We show that this ratio increases by a small factor in 
each iteration, and that it has to increase sufficiently for the matrix to be
$\eps$-balanced. This is summarized in the following two lemmas.
\begin{lemma}[change in ratio]\label{lemma:change}
$\displaystyle\frac{a^{(t+1)}_{32}}{a^{(t+1)}_{23}}\le \Big(\frac{1+7\sqrt{\beta}}{1+\eps}\Big)\cdot\frac{a^{(t)}_{32}}{a^{(t)}_{23}}$.
\end{lemma}

\begin{lemma}[stopping condition]\label{lemma:stopping}
If $A^{(t)}$ is $\eps$-balanced, then $\displaystyle\frac{a^{(t)}_{32}}{a^{(t)}_{23}} > \frac{1}{100}$.
\end{lemma}

Before proving the two lemmas we show how they lead to the
proof of Theorem~\ref{thm: lower main}.
\begin{proof}[Proof of Theorem~\ref{thm: lower main}]
By Lemma~\ref{lemma:change},
$\frac{a^{(t+1)}_{32}}{a^{(t+1)}_{23}}\le 
\left(\frac{1+7\sqrt{\beta}}{1+\eps}\right)^t\cdot\frac{a_{32}}{a_{23}} =
\left(\frac{1+7\sqrt{\beta}}{1+\eps}\right)^t\cdot \frac{\eps}{\beta + \eps}$.
By Lemma~\ref{lemma:stopping}, if $A^{(t+1)}$ is $\eps$-balanced, then
$\frac{1}{100} < \frac{a^{(t+1)}_{32}}{a^{(t+1)}_{23}}\le 
\left(\frac{1+7\sqrt{\beta}}{1+\eps}\right)^t\cdot \frac{\eps}{\beta + \eps}\le
\left(1+7\sqrt{\beta}\right)^t\cdot \frac{\eps}{\beta + \eps}$.
Using $\beta=100\eps$, we get the condition that $(1+7\sqrt{\beta})^t > \frac{101}{100}$,
which implies that $t = \Omega(1/\sqrt{\eps})$.
\end{proof}

\begin{proof}[Proof of Lemma~\ref{lemma:change}]
Using the notation we defined earlier, we have that
$f(\xt{1})  = \sum_{i,j=1}^4 a_{ij} = 4 + 2\eps + \beta$ and
$f(\x^{*}) = \sum_{i,j=1}^4 a^*_{ij} = 4 + 2\sqrt{\eps(\beta+\eps)}$,
so $f(\xt{1}) - f(\x^*) < \beta$.
We observe that at each iteration $t$, $a^{(t)}_{12} a^{(t)}_{21}=a^{(t)}_{34}a^{(t)}_{43}=1$ and $a^{(t)}_{23}a^{(t)}_{32}=\eps(\beta + \eps)$ because the product of weights of arcs on any cycle in $G_A$ is preserved (for instance, arcs $(1,2)$ and $(2,1)$ form a cycle and initially $a_{12} a_{21}=1$). 

The ratio $a^{(t)}_{32}/a^{(t)}_{23}$ is only affected in iterations that balance index $2$ or $3$. Let's assume a balancing operation at index $2$, a similar analysis applies to balancing at index $3$. By balancing at index $2$ at time $t$ we have
\begin{equation}\label{eq9}
\frac{a^{(t+1)}_{32}}{a^{(t)}_{32}} = \frac{a^{(t)}_{23}}{a^{(t+1)}_{23}}=\sqrt{\frac{a^{(t)}_{21}+a^{(t)}_{23} }{a^{(t)}_{12} + a^{(t)}_{32}}}.
\end{equation}
Thus, to prove Lemma~\ref{lemma:change}, it suffices to show that
\begin{equation}
\frac{a^{(t+1)}_{32}}{a^{(t)}_{32}} \cdot 
\frac{a^{(t)}_{23}}{a^{(t+1)}_{23}} = \frac{a^{(t)}_{21}+a^{(t)}_{23} }{a^{(t)}_{12} + a^{(t)}_{32}} \le \frac{1+7\sqrt{\beta}}{1+\eps}.
\end{equation}
By our previous observation, $a^{(t)}_{12} a^{(t)}_{21}=1$, so if $a^{(t)}_{21}=y$, then $a^{(t)}_{12}=1/y$. Similarly $a^{(t)}_{23}a^{(t)}_{32}=\eps(\beta + \eps)$ implies that there exists $z$ such that $a^{(t)}_{23}=(\beta+\eps)z$ and $a^{(t)}_{32}=\eps/z$. Therefore:
\begin{equation}\label{eq8}
\frac{a^{(t)}_{21}+a^{(t)}_{23} }{a^{(t)}_{12} + a^{(t)}_{32}} = 
{\frac{y + (\beta +\eps)z}{\displaystyle(1/y) + (\eps/z)}}
\end{equation}

We bound the right hand side of Equation~\eqref{eq8} by proving 
upper bounds on $y$ and $z$. We first show that
$y < 1 + 2\sqrt{\beta}$.
%
To see this notice that on the one hand,
\begin{equation}\label{eq6}
f(\xt{t}) = \sum_{i,j=1}^4 a^{(t)}_{ij} =  a^{(t)}_{12} + a^{(t)}_{21} + a^{(t)}_{23} + a^{(t)}_{32} + a^{(t)}_{34} + a^{(t)}_{43} \ge y + \frac{1}{y} + 2\sqrt{\eps(\beta + \eps)} + 2,
\end{equation}
where we used $a^{(t)}_{34} + a^{(t)}_{43} \ge 2$ and 
$a^{(t)}_{34} a^{(t)}_{43} \ge 2\sqrt{\eps(\beta + \eps)}$, 
both implied by the arithmetic-geometric mean inequality.
On the other hand, 
\begin{equation}\label{eq7}
f(\xt{t}) \le f(\xt{1}) \le f(\x^*) + \beta = 4 + 2\sqrt{\eps(\beta + \eps)} + \beta.
\end{equation}
Combining Equations~\eqref{eq6} and~\eqref{eq7} together, we have 
$y + ({1}/{y})  - 2 \le \beta$. For sufficiently small $\eps$, the last inequality 
implies, in particular, that $y < 2$. Thus, we have 
$(y-1)^2 \le y\beta < 2\beta$, and this implies that $y < 1 + 2\sqrt{\beta}$.

Next we show that
%
$z \le 1$.
%
Assume for contradiction that $z >1$. By the arithmetic-geometric mean 
inequality $a^{(t)}_{12} + a^{(t)}_{21}\ge 2$ and $a^{(t)}_{34} + a^{(t)}_{43}\ge 2$. Thus,
\begin{eqnarray*}
f(\xt{t})   =  \sum_{i,j=1}^4 a^{(t)}_{ij} 
 \ge  2 + (\beta + \eps)z + \frac{\eps}{z} + 2
 = 4 + \beta z + \eps\left(z + \frac{1}{z}\right)
 >  4 + \beta + 2\eps = f(\xt{1}),
\end{eqnarray*}
where the last inequality follows because $z > 1$, and $z+ 1/z > 2$.
But this is a contradiction, because each balancing iteration reduces 
the value of $F$, so $f(\xt{t})\le f(\xt{1})$.

We can now bound $({a^{(t)}_{21}+a^{(t)}_{23} })/({a^{(t)}_{12} + a^{(t)}_{32}})$. By
Equation~\eqref{eq8}, and using 
our bounds for $y$ and $z$,
$$
\frac{a^{(t)}_{21}+a^{(t)}_{23} }{a^{(t)}_{12} + a^{(t)}_{32}}
 = {\frac{y + (\beta +\eps)z}{(1/y) + (\eps/z)}}
\le {\frac{(1+2\sqrt{\beta}) + (\beta +\eps)}{\displaystyle\frac{1}{1 + 2\sqrt{\beta}} + {\eps}}}
\le  {\frac{1+4\sqrt{\beta}}{\displaystyle\frac{1}{1 + 2\sqrt{\beta}} + \frac{\eps}{1+2\sqrt{\beta}}}}
\le {\frac{1+7\sqrt{\beta}}{1+\eps}}.
$$
The last line uses the fact that $\sqrt{\beta}\gg \beta=100\eps\ge \eps$, 
which holds if $\eps$ is sufficiently small.
\end{proof}

\begin{proof}[Proof of Lemma~\ref{lemma:stopping}] 
Let $t-1$ be the last iteration before an $\eps$-balanced matrix is obtained. 
We argued that there is $z\le1$ such that $a^{(t)}_{23}=(\beta + \eps)z$ 
and $a^{(t)}_{32}=\eps/z$. Assume for the sake of contradiction that 
${a^{(t)}_{32}}/{a^{(t)}_{23}}< 1/100$. This implies that 
$(\eps/z)/((\beta + \eps)z) < 1/100$, and thus $z^2 > 100/101$. So, we get
\begin{eqnarray}
\nonumber
f(\xt{t}) - f(\x^*) &\ge& a^{(t)}_{23} + a^{(t)}_{32} - 2\sqrt{a^{(t)}_{23}a^{(t)}_{32}} = \left(\sqrt{a^{(t)}_{23}}-\sqrt{a^{(t)}_{32}}\right)^2 \ge a^{(t)}_{23}\Big(1-\sqrt{\frac{1}{100}}\Big)^2\\
\label{eq4}
& =  & 0.81\cdot(\beta + \eps)z \ge 
0.81\cdot(\beta + \eps)\cdot\sqrt{\frac{100}{101}}\ge 81\cdot\eps.
\end{eqnarray}
By Lemma~\ref{lemma:opt lower}, the left hand side the of above can 
be bounded as follows.
\begin{equation}\label{eq5}
f(\xt{t}) - f(\x^*) \le n \|\nabla f(\xt{t})\|_1 \le n^2 \|\nabla f(\xt{t})\|_2
\end{equation}
Note that for sufficiently small $\eps$, $f(\xt{t})\le f(\xt{1})\le 5$. 
Combining Equations~\eqref{eq4} and~\eqref{eq5}, and using 
$n=4$ and $f(\xt{t})\ge 5$, we get that
\begin{equation}
\frac{\|\nabla f(\xt{t})\|_2}{f(\xt{t})} > \frac{81}{80}\cdot \eps > \eps.
\end{equation}
By Equation~\eqref{eq:stop condition}, this contradicts our assumption 
that $t-1$ is the last iteration. 
\end{proof}

\section{Proofs}\label{proofs}

\begin{proof}[Proof of Lemma~\ref{lemma:progress}]
The value $f(\xt{t})$ is the sum of the entries of $A^{(t)}$. 
By Lemma~\ref{lemma:reduction}, balancing the $i$-th index 
of $A^{(t)}$ reduces the value of $f(\xt{t})$ by 
$\left(\sqrt{\|a_{.,i}^{(t)}\|_1} - \sqrt{\|a_{i,.}^{(t)}\|_1}\right)^2$. 
To simplify notation, we drop the superscript $t$ in the following 
equations. We have
\begin{equation}
\left(\sqrt{\ali{i}} - \sqrt{\alo{i}}\right)^2 = 
\frac{\left(\ali{i} - \alo{i}\right)^2}{\left(\sqrt{\ali{i}} + \sqrt{\alo{i}}\right)^2} \ge 
\frac{\left(\ali{i} - \alo{i}\right)^2}{2\left(\ali{i} + \alo{i}\right)}. \label{eq1}       
\end{equation}
It is easy to see that
\begin{equation}
\max_{i\in[n]} \frac{\left(\ali{i} - \alo{i}\right)^2}{\left(\ali{i} + \alo{i}\right)} \ge
\frac{\sum_{i=1}^n\left(\ali{i} - \alo{i}\right)^2}{\sum_{i=1}^n \left(\ali{i} + \alo{i}\right)}.\label{eq2}
\end{equation}
But the right hand side of the above inequality (after resuming the 
use of the superscript $t$) equals $\frac{\|\nabla f(\xt{t})\|_2^2}{2f(\xt{t})}$. 
This is because for all $i$, $\left(\|a_{i,.}^{(t)}\|_1-\|a_{.,i}^{(t)}\|_1\right)$ is by
Equation~\eqref{eq:partial} the $i$-th coordinate of $\nabla f(\xt{t})$, 
and in the denominator 
$\sum_{i=1}^n \left(\|a_{i,.}^{(t)}\|_1 + \|a_{.,i}^{(t)}\|_1\right) = 2f(\xt{t})$. 
Together with Equations~\eqref{eq1} and~\eqref{eq2}, this implies that 
balancing 
$i_t=\argmax_{i\in[n]}\left\{\left(\sqrt{\|a_{.,i}^{(t)}\|_1} - 
\sqrt{\|a_{i,.}^{(t)}\|_1}\right)^2\right\}$
decreases $f(\xt{t})$ by the claimed value. 
\end{proof}

\begin{proof}[Proof of Corollary~\ref{cor:progress}]
From Equation~\eqref{eq:stop condition}, we know that the diagonal 
matrix $\diag(e^{x_1},\ldots,e^{x_n})$ balances $A$ with relative error 
$\eps$ if and only if
$\displaystyle\frac{\|\nabla f(\x)\|_2}{f(\x)} \le \eps$.
Thus, if $A^{(t)}$ is not $\eps$-balanced, 
$\displaystyle\frac{\|\nabla f(\xt{t})\|_2}{f(\xt{t})} > \eps$. 
By Lemma~\ref{lemma:progress},
$f(\xt{t})-f(\xt{t+1)}) \ge \frac{\|\nabla f(\xt{t})\|_2^2}{4f(\xt{t})} = 
\frac{1}{4}\cdot\left(\frac{\|\nabla f(\xt{t})\|_2}{f(\xt{t})}\right)^2\cdot f(\xt{t}) \ge 
\frac{\eps^2}{4}\cdot f(\xt{t})$.
\end{proof}

\begin{proof}[Proof of Theorem~\ref{thm:original}]
In the original Osborne-Parlett-Reinsch algorithm, the indices are balanced 
in a fixed round-robin order. A \emph{round} of balancing is a sequence of 
$n$ balancing operations where each index is balanced exactly once. Thus, 
in the OPR algorithm all $n$ indices are balanced in the same order every 
round. We prove a more general statement that any algorithm that balances 
indices in rounds (even if the indices are not balanced in the same order 
every round) obtains an $\eps$-balanced matrix in at most $O((n\log w)/\eps^2)$ 
rounds. To this end, we show that applying a round of balancing to a matrix 
that is not $\eps$-balanced reduces the value of function $f$ at least by a 
factor of $1-{\eps^2}/{16n}$.

To simplify notation, we consider applying a round of balancing to the initial 
matrix $A^{(1)} = A$. The argument clearly holds for any time-$t$ matrix 
$A^{(t)}$. If $A$ is not $\eps$-balanced, by Lemma~\ref{lemma:progress} 
and Corollary~\ref{cor:progress}, there exists an index $i$ such that by 
balancing $i$ the value of $f$ is reduced by:
\begin{equation}\label{eq18}
f(\xt{1})-f(\xt{2}) = \left(\sqrt{\|a_{.,i}\|_1} - \sqrt{\|a_{i,.}\|_1}\right)^2 \ge 
\frac{\eps^2}{4}f(\xt{1}).
\end{equation}

If $i$ is the first index to balance in the next round of balancing, then in 
that round the value of $f$ is reduced at least by a factor of 
$1-{\eps^2}/{4} \ge 1-{\eps^2}/{16n}$, and we are done. Consider the 
graph $G_A$ corresponding to the matrix $A$. If node $i$ is not the 
first node in $G_A$ to be balanced, then some of its neighbors in the 
graph $G_A$ might be balanced before $i$. The main problem is that 
balancing neighbors of $i$ before $i$ may reduce the imbalance of 
$i$ significantly, so we cannot argue that when we reach $i$ and
balance it the value of $f$ reduces significantly. Nevertheless, we show 
that balancing $i$ and its neighbors in this round will reduce the value 
of $f$ by at least the desired amount. Let $t$ denote the time that 
$i$ is balanced in the round. For every arc $(j,i)$ into $i$, let 
$\delta_{j} = |a_{ji}-a^{(t)}_{ji}|$, and for every arc $(i,j)$ out of $i$ 
let $\sigma_{j}=|a_{ij}-a^{(t)}_{ij}|$. These values measure the weight 
change of these arcs due to balancing a neighbor of $i$ at any time 
since the beginning of the round. The next lemma shows if the weight 
of an arc incident on $i$ has changed since the beginning of the round, 
it must have reduced the value of $f$.

\begin{claim}\label{claim:neighbors}
If balancing node $j$ changes $a_{ji}$ to $a_{ji} + \delta$, then the 
balancing reduces the value of $f$ by at least $\delta^2/a_{ji}$. 
Similarly if balancing node $j$ changes $a_{ij}$ to $a_{ij} +\delta$, 
then the balancing reduces the value of $f$ by at least $\delta^2/a_{ij}$. 
\end{claim}

\begin{proof}
To simplicity notation we assume that $j$ is balanced in the first 
iteration of the round. If balancing $j$ changes $a_{ji}$ to 
$a_{ji} + \delta$, then by the definition of balancing,
\begin{equation}
\frac{a_{ji}+\delta}{a_{ji}} = \sqrt{\frac{\|a_{.,j}\|_1}{\|a_{j,.}\|_1}}.
\end{equation}
Thus, by Lemma~\ref{lemma:reduction} the value of $f$ reduces by
\begin{equation*}
\left(\sqrt{\ali{j}}-\sqrt{\alo{j}}\right)^2 = \left(\sqrt{\frac{\ali{j}}{\alo{j}}} -1\right)^2\alo{j} = \left(\displaystyle\frac{a_{ji}+\delta}{a_{ji}}-1\right)^2\alo{j} =
\left(\frac{\delta}{a_{ji}}\right)^2\alo{j} \ge \frac{\delta^2}{a_{ji}}
\end{equation*}
The proof for the second part of the claim is similar.
\end{proof}

Going back to the proof of Theorem~\ref{thm:original},
let $t$ denote the iteration in the round that $i$ is balanced. 
By Claim~\ref{claim:neighbors}, balancing neighbors of $i$ 
has already reduced the value of $f$ by 
\begin{equation}
\sum_{j:(j,i)\in E}\frac{\delta_j^2}{a_{ji}} + 
\sum_{j:(i,j)\in E}\frac{\sigma_j^2}{a_{ij}}.
\end{equation}
Balancing $i$ reduces value of $f$ by an additional 
$\left(\sqrt{\|a^{(t)}_{.,i}\|_1} - \sqrt{\|a^{(t)}_{i,.}\|_1}\right)^2$, so 
the value of $f$ in the current round is reduced by at least:
\begin{equation*}
R=\sum_{j:(j,i)\in E}\frac{\delta_j^2}{a_{ji}} + 
\sum_{j:(i,j)\in E}\frac{\sigma_j^2}{a_{ij}} + 
\left(\sqrt{\|a^{(t)}_{.,i}\|_1} - \sqrt{\|a^{(t)}_{i,.}\|_1}\right)^2
\end{equation*}
Assume without loss of generality that $\alo{i} > \ali{i}$. To lower 
bound $R$, we consider two cases:
\begin{description}
\item[case (i)] $\displaystyle\sum_{j:(j,i)\in E} \delta_j + 
\sum_{j:(i,j)\in E} \sigma_j \ge \frac{1}{2}(\alo{i}-\ali{i})$.
In this case,
\begin{align}
R \ge \sum_{j:(j,i)\in E}\frac{\delta_j^2}{a_{ji}} + 
\sum_{j:(i,j)\in E}\frac{\sigma_j^2}{a_{ij}} &\ge \frac{1}{\ali{i}}\sum_{j:(j,i)\in E}\delta_j^2 + \frac{1}{\alo{i}}\sum_{j:(i,j)\in E}\sigma_j^2\nonumber\\ 
&\ge  \frac{1}{n\ali{i}}(\sum_{j:(j,i)\in E}\delta_j)^2 + \frac{1}{n\alo{i}}(\sum_{j:(i,j)\in E}\sigma_j)^2, \label{eq10}
\end{align}
where the last inequality follows by Cauchy-Schwarz inequality. 
By assumption of case (i),
\begin{equation}\label{eq11}
\max(\sum_{j:(j,i)\in E} \delta_j, \sum_{j:(i,j)\in E}\sigma_j) \ge \frac{1}{4}(\alo{i}-\ali{i})
\end{equation}
Equations~\eqref{eq10} and~\eqref{eq11} together imply that
\begin{align*}
R &\ge \frac{(\sum_{j:(j,i)\in E}\delta_j)^2 + (\sum_{j:(i,j)\in E}\sigma_j)^2}{n\max(\ali{i}, \alo{i})} \ge \frac{1}{16n}\frac{(\alo{i}-\ali{i})^2}{\max(\ali{i}, \alo{i})} \\
&= \frac{\left(\sqrt{\|a_{.,i}\|_1}-\sqrt{\|a_{i,.}\|_1}\right)^2\left(\sqrt{\|a_{.,i}\|_1}+\sqrt{\|a_{i,.}\|_1}\right)^2}{16n\max(\ali{i}, \alo{i})} \ge \frac{1}{16n}\left(\sqrt{\|a_{.,i}\|_1}-\sqrt{\|a_{i,.}\|_1}\right)^2.
\end{align*}

\item[case (ii)] $\displaystyle\sum_{j:(j,i)\in E} \delta_j + \sum_{j:(i,j)\in E} \sigma_j < \frac{1}{2}(\alo{i}-\ali{i})$. By definition of $\delta_j$'s and $\sigma_j$'s:
\begin{align}
\ali{i} - \sum_{j:(j,i)\in E} \delta_j \le \|a^{(t)}_{.,i}\|_1 \le \ali{i}  + \sum_{j:(j,i)\in E} \delta_j\label{eq20}\\
\alo{i} - \sum_{j:(i,j)\in E} \sigma_j \le \|a^{(t)}_{i,.}\|_1 \le \alo{i}  +\sum_{j:(i,j)\in E} \sigma_j\label{eq21}.
\end{align}
Combining Equations~\eqref{eq20} and~\eqref{eq21}, and the 
assumption of case (ii) gives:
\begin{align}
\|a^{(t)}_{i,.}\|_1 + \|a^{(t)}_{.,i}\|_1 \le  \alo{i} +  \ali{i} + \sum_{j:(i,j)\in E} \sigma_j  + \sum_{j:(j,i)\in E} \delta_j \le 2\left(\alo{i} + \ali{i}\right)\label{eq22}\\
\|a^{(t)}_{i,.}\|_1 - \|a^{(t)}_{.,i}\|_1 \ge \alo{i}  - \ali{i} - \sum_{j:(i,j)\in E} \sigma_j  - \sum_{j:(j,i)\in E} \delta_j \ge \frac{1}{2}\left(\alo{i} - \ali{i}\right).\label{eq23}
\end{align}
Using Equations~\eqref{eq22} and~\eqref{eq23}, we can write:
\begin{align*}
R \ge \left(\sqrt{\|a^{(t)}_{i,.}\|_1} - \sqrt{\|a^{(t)}_{.,i}\|_1}\right)^2 &= \frac{\left(\|a^{(t)}_{.,i}\|_1 - \|a^{(t)}_{i,.}\|_1 \right)^2}{\left(\sqrt{\|a^{(t)}_{.,i}\|_1} + \sqrt{\|a^{(t)}_{i,.}\|_1}\right)^2} \ge \frac{\left(\alo{i} -\ali{i}\right)^2}{8\left(\|a^{(t)}_{i,.}\|_1 + \|a^{(t)}_{.,i}\|_1\right)}\\
&\ge \frac{\left(\alo{i} -\ali{i}\right)^2}{16\left(\alo{i} + \ali{i}\right)} \ge \frac{1}{16}\left(\sqrt{\alo{i}} -\sqrt{\ali{i}}\right)^2.
\end{align*}
\end{description}
Thus, we have shown in both cases that in one round the balancing 
operations on node $i$ and its neighbors reduces the value of $f$ by 
at least 
\begin{equation}
\frac{1}{16n}\left(\sqrt{\|a_{.,i}\|_1}-\sqrt{\|a_{i,.}\|_1}\right)^2,
\end{equation}
which in turn is at least $\Omega(\frac{\eps^2}{n}f(\xt{1}))$ by 
Equation~\eqref{eq18}. Thus, we have shown that if $A$ is not 
$\eps$-balanced, one round of balancing (where each index is 
balanced exactly once) reduces the objective function $f$ by a 
factor of at least $1-\Omega\left(\frac{\eps^2}{n}f(\xt{1})\right)$. 
By an argument similar to the one in the proof of 
Theorem~\ref{thm:greedy}, we get that the algorithm obtains 
an $\eps$-balanced matrix in at 
most $O(\eps^{-2}n\log w)$ rounds. The number of balancing
iterations in each round is $n$, and the number of arithmetic 
operations in each round is $O(m)$, so the original OPR algorithm 
obtains an $\eps$-balanced matrix using 
$O(\eps^{-2}mn\log w)$ arithmetic operations.
\end{proof}

\begin{algorithm}[t]
\caption{RandomBalance($A$, $\eps$)}\label{alg:random}
\begin{algorithmic}[1]
 \Input{Matrix $A \in\mathbb{R}^{n\times n}, \eps$}
\Output{An $\eps$-balanced matrix}\vspace{2mm}

\State $r = a_{\min}\cdot({\eps}/{wn})^{10}$
\State Let $\hat{a}_{ij}^{(1)} = a_{ij}^{(1)}$ for all $i$ and $j$
\State Let $\hao = \|a_{i,.}^{(t)}\|$ and $\hai = \|a_{.,i}^{(t)}\|$ for all $i$
\For{$t= 1$ to $O(\eps^{-2}\ln w)$}
\State Pick $i$ randomly with probability $p_i = \frac{\hao + \hai}{2\sum_{i,j}\haij}$
\If {$\hao + \hai \ge \eps a_{\min} / 10 w n$}
\State $\hat{M}_i = \max\{\hao, \hai\}$, $\hat{m}_i = \min\{\hao, \hai\}$
\If{$\hat{m}_i = 0$ {\bf or} $\hat{M}_i/\hat{m}_i \ge 1 + \eps/n$}
\If{$\hat{m}_i \neq 0$}            $\alpha = \frac{1}{2}\ln(\hai/\hao)$
\ElsIf{$\hat{m}_i = \hai = 0$}  $\alpha = \frac{1}{2}\ln(nr/\hao)$
\ElsIf{$\hat{m}_i = \hao = 0$} $\alpha = \frac{1}{2}\ln(\hai/nr)$
\EndIf
\State Let $\hxt{t+1}\gets\hxt{t} + \alpha\mathbf{e}_i$ (truncated to 
          $O(\ln(wn/\eps))$ bits of precision)
\For{$j=1$ to $n$}
 \If{$j$ is a neighbor of $i$}
 \State $\widehat{a}_{ij}^{(t+1)} \gets a_{ij}e^{\hx_i^{(t+1)}- \hx_j^{(t+1)}}$ 
 and $\widehat{a}_{ji}^{(t+1)} \gets a_{ji}e^{\hx_j^{(t+1)}- \hx_i^{(t+1)}}$,
 (truncated to $O(\ln(wn/\eps))$ bits)
 \State $\|\widehat{a}_{j,.}^{(t+1)}\| = \|\widehat{a}_{j,.}^{(t)}\| - \haji + \widehat{a}_{ji}^{(t+1)}$
 and $\|\widehat{a}_{.,j}^{(t+1)}\| = \|\widehat{a}_{.,j}^{(t)}\| - \haij + \widehat{a}_{ij}^{(t+1)}$
 \EndIf
\EndFor
\State  $\|\widehat{a}_{i,.}^{(t+1)}\| = \sum_{j=1}^n \widehat{a}_{ij}^{(t+1)}$ 
           and $\|\widehat{a}_{.,i}^{(t+1)}\|= \sum_{j=1}^n \widehat{a}_{ji}^{(t+1)}$
\EndIf
\EndIf
\EndFor
\State\textbf{return} the resulting matrix
\end{algorithmic}
\end{algorithm}

\begin{proof}[Proof of Lemma~\ref{lemma:case}]
Notice that for every $i$,
$$
(\hao + \hai)\cdot \left(\sqrt{\|a_{i,.}^{(t)}\|}-\sqrt{\|a_{.,i}^{(t)}\|}\right)^2
\le (\hao + \hai) \cdot 
\frac{\left({\|a_{i,.}^{(t)}\|}-{\|a_{.,i}^{(t)}\|}\right)^2}{\|a_{i,.}^{(t)}\| + \|a_{.,i}^{(t)}\|}\le
\left(\|a_{i,.}^{(t)}\|-\|a_{.,i}^{(t)}\|\right)^2,
$$
because $\hao + \hai\le \|a_{i,.}^{(t)}\| + \|a_{.,i}^{(t)}\|$.
We first bound the sum over $i\notin A$.
\begin{eqnarray*}
           \sum_{i\notin A} p_i\cdot\left(\sqrt{\|a_{i,.}^{(t)}\|}-\sqrt{\|a_{.,i}^{(t)}\|}\right)^2 
& = &  \sum_{i\notin A} \frac{\hao + \hai}{2\sum_{i,j}\haij}\cdot
           \left(\sqrt{\|a_{i,.}^{(t)}\|}-\sqrt{\|a_{.,i}^{(t)}\|}\right)^2  \\ 
&\le&  \frac{1}{2\sum_{i,j}\haij}\cdot\sum_{i\notin A}
           \left(\|a_{i,.}^{(t)}\|-\|a_{.,i}^{(t)}\|\right)^2 \\
&\le&  \frac{1}{2\sum_{i,j}\haij}\cdot\sum_{i\notin A}
           \left(\|\widehat{a}_{i,.}^{(t)}\|+\|\widehat{a}_{.,i}^{(t)}\| + 2nr\right)^2 \\
&\le&  \frac{1}{2\sum_{i,j}\haij}\cdot \sum_{i\notin A} (2\eps a_{\min}/10wn)^2 \\
&\le&  \frac{1}{2\sum_{i,j}\haij}\cdot n\cdot (\eps a_{\min}/5wn)^2 \le 
            \frac{\eps^2}{25n}\cdot a_{\min}\le \frac{\eps^2}{25n}\cdot f(\xt{t})
\end{eqnarray*}
where the second inequality follows because,
for every $j$, $a_{ij}^{(t)} \le \widehat{a}_{ij}^{(t)} + r$ and 
$a_{ji}^{(t)} \le \widehat{a}_{ji}^{(t)} + r$, and the third inequality
follows because $\hao + \hai < \eps a_{\min}/10wn$ and 
$nr < \eps a_{\min}/10wn$.

Next, we bound the sum over $i\in A\setminus (B\cup C)$.
Recall $\hat{M}_i = \max\{\hao, \hai\}$ and 
$\hat{m}_i = \min\{\hao, \hai\}$. 
Put $M_i = \max\{\|a_{i,.}^{(t)}\|,\|a_{.,i}^{(t)}\|\}$ and 
$m_i = \min\{\|a_{i,.}^{(t)}\|, \|a_{.,i}^{(t)}\|\}$. Let
$k = \arg\max_{i\in A\setminus (B\cup C)} (M_i - m_i)^2$.
We have
\begin{eqnarray*}
         \sum_{i\in A\setminus (B\cup C)} p_i\cdot
         \left(\sqrt{\|a_{i,.}^{(t)}\|}-\sqrt{\|a_{.,i}^{(t)}\|}\right)^2 
&\le& \frac{1}{2\sum_{i,j}\haij}\cdot 
          \sum_{i\in A\setminus(B\cup C)}\left({\|a_{i,.}^{(t)}\|}-\|a_{.,i}^{(t)}\|\right)^2 \\
&\le& \frac{1}{2\sum_{i,j}\haij}\cdot \sum_{i\in A\setminus(B\cup C)}\left(M_i - m_i\right)^2 \\
&\le& \frac{1}{2\sum_{i,j}\haij}\cdot n\cdot m_k^2\left(\frac{M_k}{m_k}-1\right)^2.
\end{eqnarray*}
To bound the last quantity, we prove an upper bound on $\frac{M_k}{m_k}$ 
using the fact that $\frac{\hat{M}_k}{\hat{m}_k} < 1 + \frac{\eps}{n}$. As 
$k\in A$, we have $\hat{M}_k + \hat{m}_k = \khao + \khai \ge 
\frac{\eps a_{\min}}{10wn}$. Thus, 
$\hat{M}_k \ge \frac{\eps a_{\min}}{20wn}$. Combining this with 
$\frac{\hat{M}_k}{\hat{m}_k} < 1 + \frac{\eps}{n}$ implies that 
$\hat{m}_k > \frac 1 2 \hat{M}_k \ge \frac{\eps a_{\min}}{40wn}$.
Hence,
$$
\frac{M_k}{m_k} \le \frac{M_k}{\hat{m}_k} \le 
\frac{\hat{M}_k + nr}{\hat{m}_k} \le \frac{\hat{M}_k}{\hat{m}_k} + 
\frac{nr}{\eps a_{\min}/40wn} =  \frac{\hat{M}_k}{\hat{m}_k} + 
{40n}\cdot\left(\frac{\eps }{wn}\right)^9 \le 
\frac{\hat{M}_k}{\hat{m}_k} + \frac{40\eps^9}{n^8} \le 1 + \frac{2\eps}{n}.
$$
(Notice that $w\ge 1$.)
Using the upper bound on $\frac{M_k}{m_k}$, we obtain
\begin{eqnarray*}
                   \sum_{i\in A\setminus (B\cup C)} 
                       p_i\cdot\left(\sqrt{\|a_{i,.}^{(t)}\|}-\sqrt{\|a_{.,i}^{(t)}\|}\right)^2
&\le& \frac{1}{2\sum_{i,j}\haij}\cdot n\cdot m_k^2\left(\frac{M_k}{m_k}-1\right)^2 \\
&\le&  \frac{1}{2\sum_{i,j}\haij}\cdot n\cdot m_k^2\left(\frac{2\eps}{n}\right)^2 \\
&\le& \frac{2\eps^2}{n}\cdot m_k \le \frac{\eps^2}{n}\cdot f(\xt{t}),
\end{eqnarray*}
where the penultimate inequality uses the fact that 
$m_k\le \hat{m}_k + nr\le \hat{m}_k + \frac{\eps a_{\min}}{40wn} 
< 2\hat{m}_k \le \hat{M}_k + \hat{m}_k\le\sum_{i,j}\haij$.
\end{proof}

\begin{proof}[Proof of Lemma~\ref{lemma:decrease}]
We will assume that $\eps < \frac{1}{10}$.
We first consider the case that $i\in A\cap B$
(notice that $B\cap C = \emptyset$).
The update using $O(\ln(wn/\eps))$ bits of precision gives
$\widehat{x}^{(t)}_i + \alpha - (\eps/nw)^{10} \le 
\widehat{x}^{(t+1)}_i \le \widehat{x}^{(t)}_i + \alpha$,
so 
$$
\sqrt{\frac{\hai}{\hao}}\cdot e^{\widehat{x}^{(t)}_i  - (\eps/wn)^{10}}
\le e^{\widehat{x}^{(t+1)}_i} \le \sqrt{\frac{\hai}{\hao}}\cdot e^{\widehat{x}^{(t)}_i}.
$$
Therefore,
$$
\|a_{i,.}^{(t+1)}\| = 
\sum_{j=1}^n a_{ij}e^{\widehat{x}_i^{(t+1)} - \widehat{x}_j^{(t+1)}} \le 
\sqrt{\frac{\hai}{\hao}}\cdot \sum_{j=1}^n a_{ij}e^{\widehat{x}_i^{(t)} - \widehat{x}_j^{(t)}} = 
\sqrt{\frac{\hai}{\hao}}\cdot \|a_{i,.}^{(t)}\|\sqrt{\frac{\hai}{\hao}},
$$
and
$$
\|a_{.,i}^{(t+1)}\| = 
\sum_{j=1}^n a_{ji}e^{\widehat{x}_j^{(t+1)} - \widehat{x}_i^{(t+1)}} \le 
e^{(\eps/wn)^{10}}\cdot \sqrt{\frac{\hao}{\hai}}\cdot 
    \sum_{j=1}^n a_{ji}e^{\widehat{x}_j^{(t)} - \widehat{x}_i^{(t)}} \le
(1+2(\eps/wn)^{10}))\cdot \sqrt{\frac{\hao}{\hai}}\cdot \|a_{.,i}^{(t)}\|.
$$
We used the fact that $e^x \le 1+2x$ for $x\le \frac 1 2$.
We will now use the notation $\hat{M}_i$, $\hat{m}_i $, $M_i$,
and $m_i$ (the reader can recall the definitions from the proof
of Lemma~\ref{lemma:case}). We also put $\delta = 2(\eps/wn)^{10}$, 
and $\sigma = \frac{\hat{M}_i/\hat{m}_i}{{M_i}/{m_i}}$.
Thus, decrease of function $f(\cdot)$ due to balancing $i$ is
$f(\hxt{t}) - f(\hxt{t+1}) = M_i + m_i - \|a_{.,i}^{(t+1)}\| - \|a_{i,.}^{(t+1)}\|
\ge M_i + m_i - (1+\delta)\left(M_i\sqrt{\hat{m}_i/\hat{M}_i} + 
             m_i\sqrt{\hat{M}_i/\hat{m}_i}\right)
= M_i + m_i - (1+\delta)\left(\sqrt{1/\sigma} + \sqrt{\sigma}\right)\cdot \sqrt{M_i m_i}
= \left(\sqrt{M_i} -\sqrt{m_i}\right)^2 - 
    \left((1+\delta)/\sqrt{\sigma} + (1+\delta)\sqrt{\sigma} - 2 \right)\cdot \sqrt{M_i m_i}$.
We now consider three cases, and in each case show that 
$$
\left((1+\delta)/\sqrt{\sigma} + (1+\delta)\sqrt{\sigma} - 2 \right)\cdot \sqrt{M_i m_i} 
\le \frac{9}{10}\cdot \left(\sqrt{M_i} -\sqrt{m_i}\right)^2.
$$ 

\paragraph{case (i):} 
$1\le \sigma < 1 + \frac{\eps^4}{n^2}$. We first note that
$M_i\ge \hat{M}_i\ge \frac{\hat{M}_i + \hat{m}_i}{2} > \frac{\eps a_{\min}}{20wn}$.
Also, $m_i \le \hat{m}_i + nr$, so
$\frac{m_i}{M_i}\le \frac{\hat{m}_i + nr}{\hat{M}_i} \le 
\frac{1}{1+\eps/n} + \frac{nr}{\hat{M}_i} \le 1- \frac{\eps}{2n}$.
Since $\eps < \frac{1}{10}$, we have
\begin{eqnarray*}
                  \frac{9}{10}\left(1-\sqrt{\frac{m_i}{M_i}}\right)^2 
&\ge& \frac{9}{10}\left(1-\sqrt{1-\frac{\eps}{2n}}\right)^2 \\
&\ge& \frac{4\eps^4}{n^2} \\
&\ge& \left((1 + \delta) + (1+\delta)\cdot \left(1+\frac{\eps^4}{n^2}\right) -2\right) \\
&\ge&  \left(\frac{1+\delta}{\sqrt{\sigma}} + 
              (1+\delta)\sqrt{\sigma} - 2 \right)\cdot \sqrt{\frac{m_i}{M_i}},
\end{eqnarray*}
where the third inequality holds by definition of $\delta$, and the last inequality 
holds because $m_i/M_i \le 1$ and $\sigma \in [1, 1 + \eps^4/n^2]$. By 
multiplying both sides of the inequality by $M_i$ we obtain the desired
bound.

\paragraph{case (ii):}
$\sigma < 1$. We first prove a lower bound on the value of $\sigma$, as
follows: $\frac{\hat{M}_i}{\hat{m}_i} \ge \frac{\hat{M}_i}{m_i}\ge 
\frac{M_i - nr}{m_i}\ge\frac{M_i}{m_i}(1-\frac{nr}{M_i}) \ge 
\frac{M_i}{m_i}\cdot \left(1-\frac{nr}{M_i}\right) \ge 
\frac{M_i}{m_i}\cdot\left(1-\displaystyle\frac{20\eps^9}{n^8}\right)$,
and thus $\sigma\ge 1-\frac{20\eps^9}{n^8}$.
So we have
\begin{eqnarray*}
\left(\frac{1+\delta}{\sqrt{\sigma}} + (1+\delta)\sqrt{\sigma} - 2 \right)\cdot\sqrt{\frac{m_i}{M_i}}
&\le& \frac{1 + \delta}{\sqrt{1-\frac{20\eps^9}{n^8}}} + (1+ \delta) - 2 \\
&\le& (1+\delta)\left(1 + \frac{20\eps^9}{n^8}\right) + (1 + \delta) - 2 \\
&\le& \frac{24\eps^9}{n^8} < \frac{4\eps^4}{n^2} \le 
          \frac{9}{10}\cdot\left(1-\sqrt{\frac{m_i}{M_i}}\right)^2,
\end{eqnarray*} 
proving the desired inequality in this case. The first inequality holds 
because $\frac{m_i}{M_i} \le 1$ and $1-\frac{20\eps^9}{n^8}\le\sigma\le 1$.

\paragraph{case (iii):}
$\sigma > 1 + \frac{\eps^4}{n^2}$.
The idea is to show that $M_i/m_i$ is large so the desired inequality 
follows. We know that
$\frac{\sigma M_i}{m_i} = \frac{\hat{M}_i}{\hat{m}_i}\le 
\frac{M_i}{\hat{m}_i}$ and therefore $\hat{m}_i\le \frac{m_i}{\sigma}$.
On the other hand, $\hat{m}_i \ge m_i-nr$, so
$m_i\le \frac{nr}{1-1/\sigma}$. Clearly, $1/\sigma < 1-\frac{\eps^4}{2n^2}$, 
so $m_i < \frac{nr}{\eps^4/2n^2}$. Also, $M_i\ge \eps a_{\min}/20wn$.
Therefore,
$\frac{M_i}{m_i}\ge \frac{\eps a_{\min}/20wn}{2n^3r/\eps^4}\ge 
\frac{n^6}{40\eps^5}$.
Next, notice that since $\hat{m}_i > 0$ it must be that
$\hat{m}_i\ge r$. Therefore,
$m_i\le \hat{m}_i + nr\le 2n\hat{m}_i$.
This implies that $\frac{\hat{M}_i}{\hat{m}_i}\le \frac{M_i}{\hat{m}_i} \le 
2n\cdot\frac{M_i}{m_i}$, so $\sigma \le 2n$. Finally,
$$
\left(\frac{1+\delta}{\sqrt{\sigma}} + (1+\delta)\sqrt{\sigma} - 2 \right) \le 
(1+\delta)\cdot\sqrt{2n}\le 
\le \frac{1}{10}\cdot\sqrt{\frac{M_i}{m_i}},
$$
with room to spare (using the lower bound on $\frac{M_i}{m_i}$).
Multiplying both sides by $\sqrt{\frac{M_i}{m_i}}$ gives
$$
\left(\frac{1+\delta}{\sqrt{\sigma}} + (1+\delta)\sqrt{\sigma} - 2 \right)\cdot
     \sqrt{\frac{M_i}{m_i}} \le 
\frac{1}{10}\frac{M_i}{m_i}\le \frac{9}{10}\left(\sqrt{\frac{M_i}{m_i}}-1\right)^2,
$$
with more room to spare. This completes proof of the case $i\in A\cap B$.

We now move on to the case $i\in A\cap C$, so 
$\hat{M}_i + \hat{m}_i \ge \frac{\eps a_{\min}}{10wn}$ and 
$\hat{m}_i=0$. In the algorithm, $\alpha =\frac{1}{2}\ln(nr/\hao)$ or 
$\alpha = \frac{1}{2}\ln(\hai/nr)$. The idea is that we therefore replace
$\hat{m}_i$ (which is $0$) by $nr$ in some of the equations.
In particular,
$f(\hxt{t}) - f(\hxt{t+1}) \ge 
{M}_i + {m}_i - 
(1+\delta)\left(M_i\sqrt{\frac{nr}{\hat{M}_i}} + m_i\sqrt{\frac{\hat{M}_i}{nr}}\right)$.
Note that since $\hat{m}_i = 0$ then $m_i \le nr$. Therefore, 
$\frac{\hat{M}_i}{nr} \le \frac{M_i}{nr}\le \frac{M_i}{m_i}$.
On the other hand, since $i\in A$, 
$\hat{M}_i \ge \eps a_{\min}/20wn$, so
$\frac{\hat{M}_i}{nr} \ge \frac{\eps a_{\min}/20wn}{n(\eps/wn)^{10} a_{\min}} \ge 
\frac{n^8}{20\eps^9}$.
Thus we get
\begin{eqnarray*}
             f(\hxt{t}) - f(\hxt{t+1}) 
&\ge& {M}_i + {m}_i - 
           (1+\delta)\left(M_i\sqrt{\frac{nr}{\hat{M}_i}} + m_i\sqrt{\frac{\hat{M}_i}{nr}}\right) \\
&\ge& M_i + m_i - 
          (1+\delta)\left(M_i\sqrt{\frac{20\eps^9}{n^8}} + m_i\sqrt{\frac{M_i}{m_i}}\right) \\
&\ge& M_i + m_i - 2(1+\delta)M_i\sqrt{\frac{20\eps^9}{n^8}} \\
&\ge& M_i\left(1 - \frac{20\eps^4}{n^4}\right)\ge \frac{1}{10}M_i \ge 
           \frac{1}{10}(\sqrt{M_i}-\sqrt{m_i})^2,
\end{eqnarray*}
where the third inequality holds because 
$m_i\sqrt{\frac{M_i}{m_i}} = M_i\sqrt{\frac{m_i}{M_i}}\le M_i\sqrt{\frac{nr}{\hat{M}_i}}$.
\end{proof}


\bibliographystyle{plain}
\bibliography{main}

\end{document}